\documentclass[12pt]{article}

\usepackage{amssymb,amsmath,amsthm,amscd,mathrsfs,amsbsy,amsfonts}
\usepackage{pstricks,pst-node}
\usepackage{amsbsy,young,colortbl,cite}

\usepackage{graphicx}

\usepackage{amsmath}

\usepackage{amssymb,amsthm}
\usepackage[english]{babel}
\renewcommand{\d}{\partial}

\newtheorem{proposition}{Proposition}
\newtheorem{lemma}{Lemma}
\newtheorem{definition}{Definition}
\newtheorem{theorem}{Theorem}

\makeatletter      
\@addtoreset{equation}{section}

\newcommand{\gl}{M_N(\C)}

\newcommand{\I}{\mathbb{I}}
\newcommand{\bS}{\mathbb{S}}
\renewcommand{\d}{\mathrm{d}}

\newcommand{\Exp}[1]{\operatorname{e}^{#1}}

\newcommand{\g}{\mathfrak{g}}

\newcommand{\Cc}{\mathcal{C}}

\renewcommand{\L}{\mathcal L}

\newcommand{\W}{\mathcal W}

\newcommand{\R}{\mathbb R}
\newcommand{\Z}{\mathbb Z}
\newcommand{\C}{\mathbb C}
\newcommand{\N}{\mathbb N}
\def\res{\mathop{\rm Res}\nolimits}

\renewcommand{\L}{\mathcal{L}}

\renewcommand{\t}{\mathbf{t}}

\makeatletter
    
    \newcommand{\Rmnum}[1]{\expandafter\@slowromancap\romannumeral #1@}
  \makeatother
\def\res{\mathop{\rm Res}\nolimits}
\def\({\left(}
\def\){\right)}
\def\[{\begin{eqnarray}}
\def\]{\end{eqnarray}}
\def\d{\partial}
\def\ga{\alpha}

\newcommand{\La}{\Lambda}
\begin{document}

\title{On the extended multi-component  Toda hierarchy}

\author{
Chuanzhong Li\dag,\  \ Jingsong He\ddag\footnote{corresponding author}\\
 Department of Mathematics,  Ningbo University, Ningbo, 315211, China\\
\dag lichuanzhong@nbu.edu.cn\\
\ddag hejingsong@nbu.edu.cn }

\date{}

\maketitle

\abstract{
The extended flow equations of the multi-component  Toda hierarchy are constructed. We give the Hirota bilinear equations and tau function of this new extended multi-component  Toda hierarchy(EMTH). Because of logarithmic terms, some extended vertex operators are constructed in generalized Hirota bilinear equations which might be useful in topological field theory and Gromov-Witten theory. Meanwhile the Darboux transformation and bi-Hamiltonian structure of this hierarchy are given. From  the Hamiltonian tau symmetry, we give another different tau function of this hierarchy with some unknown mysterious connections with the one defined from
the point of wave functions.
}
\\
\maketitle

\section {Introduction}

The KP hierarchy and Toda lattice hierarchy
  as  completely integrable systems  have many important applications in mathematics and physics including the representation theory of Lie algebra, orthogonal polynomials and  random
matrix model  \cite{Toda,Todabook,UT,witten,dubrovin}. The KP and Toda systems have many kinds of reductions or extensions, for example the BKP, CKP hierarchy, extended Toda hierarchy (ETH)\cite{CDZ,M}, bigraded Toda hierarchy (BTH)\cite{C}-\cite{ourBlock} and so on \cite{TH14}. There are some other generalizations called multi-component KP \cite{kac,avanM} and multi-component Toda systems\cite{manasInverse2} which attract more and more attention because their widely use in many fields such as the fields of multiple orthogonal
polynomials and non-intersecting Brownian motions.

The multicomponent KP hierarchy  was discussed
with its application in representation theory and random matrix model in \cite{kac,avanM}. In
\cite{UT},  it was noticed that the $\tau$ functions of a
$2N$-multicomponent KP provide solutions of the $N$-multicomponent
2D Toda hierarchy.
The multicomponent 2D Toda hierarchy
 was considered from the point of view of the Gauss-Borel factorization problem, the theory  of multiple matrix orthogonal polynomials, non-intersecting Brownian motions and matrix Riemann-Hilbert problem \cite{manasInverse2}-\cite{manas}. In fact the multicomponent 2D Toda hierarchy in \cite{manasinverse} is a periodic reduction of the bi-infinite matrix-formed two dimensional Toda hierarchy. The coefficients(or dynamic variables) of the multicomponent 2D Toda hierarchy take values in complex finite-sized
matrices. The multicomponent 2D Toda hierarchy contains the matrix-formed Toda equation as the first flow equation.

Considering its application in the Gromov-Witten theory, the Toda hierarchy was extended to the extended Toda hierarchy\cite{CDZ} which governs the Gromov-Witten invariants of $CP^1$.
 The extended bigraded Toda
hierarchy(EBTH) is the extension of the bigraded Toda
hierarchy (BTH) which includes $N+M$ series of additional logarithmic flows\cite{C,KodamaCMP} with considering its application in the Gromov-Witten theory of orbifolds $C_{N,M}$. The Hirota bilinear equations of EBTH were equivalently constructed in our early paper\cite{ourJMP} and a very recent paper \cite{leurhirota}, because the equivalence of $t_{1,N}$ flow and $t_{0,N}$ flow of EBTH. Meanwhile it was proved to govern Gromov-Witten invariants of the total
  descendent potential of $\mathbb{P}^1$ orbifolds $C_{N,M}$ \cite{leurhirota}.
A natural question is what about the corresponding extended multi-component  Toda hierarchy( as a matrix-valued generalization of extended Toda hierarchy\cite{CDZ}) and extended multicomponent bigraded Toda hierarchy.  There is a class of orbifolds which should be governed by these generalized multicomponent
logarithmic hierarchies. That is why we think this new kind of  logarithmic hierarchy which might be useful in
Gromov-Witten invariants theory governed by these two new hierarchies. With this motivation, this paper will be partly aimed at constructing a kind of Hirota quadratic equation taking values in a  matrix-valued  differential
algebra set. This kind of Hirota bilinear equations might be useful in Gromov-Witten theory and noncommutative symplectic geometry.

This paper is arranged as follows. In the next section we will recall the factorization problem and construct the logarithmic matrix operators using which we will define the extended flows of the multicomponent bigraded Toda hierarchy. In Sections 3,
we will give the Lax equations of the extended multicomponent bigraded Toda hierarchy (EMTH), meanwhile the multicomponent Toda equations and the extended equations are introduced in this hierarchy. By Sato equations, Hirota bilinear equations of the EMTH are proved in
Section 4. The tau function of the EMTH will be defined in Section 5 which leads to the formalism of the generalized matrix-valued vertex operators and Hirota quadratic equations in Section 6. In Section 7, the multi-Darboux transformation of the EMTH is constructed using determinant techniques. After this, to prove the integrability of the EMTH, the bi-Hamiltonian structure and tau symmetry of this hierarchy are given. The last section will be devoted to a short conclusion and discussion.

\section{Factorization and logarithmic operators}

In this section, we will denote $G$ as a group which contains invertible elements of $N\times N$
complex matrices and denote its Lie algebra  $\g$ as the associative algebra  of  $N\times N$
complex matrices $M_N(\C)$.
Now we will consider the linear space of functions
$g:\R\rightarrow M_N(\C)$ with the shift operator $\Lambda$ acting on any functions $g(x)$ as
$(\Lambda g)(x):=g(x+\epsilon)$. A Left
multiplication by  $X:\R\to \gl$ is as $X\Lambda^j$, which acts on an arbitrary function $g(x)$ as $(
X\Lambda^j)(g)(x):=X(x)\cdot g(x+j\epsilon)$. Also two operators $X(x)\Lambda^i$ and $Y(x)\Lambda^j$  have the product as
$(X(x)\Lambda^i)\cdot(Y(x)\Lambda^j):=X(x)Y(x+i\epsilon)\Lambda^{i+j}.$

This Lie algebra as a linear space has the following important splitting
\begin{gather}\label{splitting}
\g=\g_+\oplus\g_-,
\end{gather}
where
\begin{align*}
  \g_+&=\Big\{\sum_{j\geq 0}X_j(x)\Lambda^j,\quad X_j(x)\in\gl\Big\},&
  \g_-&=\Big\{\sum_{j< 0}X_j(x)\Lambda^j,\quad X_j(x)\in\gl\Big\}.
\end{align*}

The splitting
\eqref{splitting} leads us to consider the following factorization of
$g\in G$
\begin{gather}\label{fac1}
g=g_-^{-1}\cdot g_+, \quad g_\pm\in G_\pm
\end{gather}
where $G_\pm$ have $\g_\pm$ as their Lie algebras. Here $G_+$
is the set of invertible linear operators  of the
form $\sum_{j\geq 0}g_j(x)\Lambda^j$; while $G_-$ is the set of
invertible linear operators in the form of
$1+\sum_{j<0}g_j(x)\Lambda^j$.

 Now we
introduce  the following free operators $ W_0,\bar  W_0\in G$
\begin{align}
 \label{def:E}  W_0&:=\sum_{k=1}^NE_{kk}\Exp{\sum_{j=0}^\infty
 t_{jk}\Lambda^{j}+ s_{j}\frac{\Lambda^{j}}{j!}(\partial-c_j)}, \\
\label{def:barE}   \bar W_0&:=\sum_{k=1}^NE_{kk}\Exp{\sum_{j=0}^\infty\bar
   t_{j k}\Lambda^{-j}+ s_{j}\frac{\Lambda^{-j}}{j!}(\partial-c_j)},
\end{align}
where $t_{jk}, \bar t_{jk},s_{j} \in \C$
will play the role of continuous times.
 We   define the dressing operators $W,\bar W$ as follows
\begin{align}
\label{def:baker}W&:=S\cdot W_0,& \bar W&:=\bar S\cdot \bar  W_0.
\end{align}
Given an element $g\in G$ and time series $t=(t_{jk}), \bar t=(\bar t_{jk}), s=(s_{j}); j,k \in \N, 1\leq k \leq N$, one can consider the factorization problem in $G$ \cite{manasinverse}
\begin{gather}
  \label{facW}
  W\cdot g=\bar W,
\end{gather}
i.e.
 the factorization problem
\begin{gather}
  \label{factorization}
  S(t,\bar t,s)\cdot W_0\cdot g=\bar S(t,\bar t,s)\cdot\bar W_0,\quad S\in G_-\text{ and } \bar S\in G_+.
\end{gather}
Observe that  $S,\bar S$ have expansions of the form
\begin{gather}
\label{expansion-S}
\begin{aligned}
S&=\I_N+\omega_1(x)\Lambda^{-1}+\omega_2(x)\Lambda^{-2}+\cdots\in G_-,\\
\bar S&=\bar\omega_0(x)+\bar\omega_1(x)\Lambda+\bar\omega_2(x)\Lambda^{2}+\cdots\in
G_+.
\end{aligned}
\end{gather}
Also we define the symbols of $S,\bar S$ as  $\bS,\bar \bS$
\begin{gather}
\begin{aligned}
\bS&=\I_N+\omega_1(x)\lambda^{-1}+\omega_2(x)\lambda^{-2}+\cdots,\\
\bar \bS&=\bar\omega_0(x)+\bar\omega_1(x)\lambda+\bar\omega_2(x)\lambda^{2}+\cdots.
\end{aligned}
\end{gather}

Also the inverse operators $S^{-1},\bar S^{-1}$ of operators $S,\bar S$ have expansions of the form
\begin{gather}
\begin{aligned}
S^{-1}&=\I_N+\omega'_1(x)\Lambda^{-1}+\omega'_2(x)\Lambda^{-2}+\cdots\in G_-,\\
\bar S^{-1}&=\bar\omega'_0(x)+\bar\omega'_1(x)\Lambda+\bar\omega'_2(x)\Lambda^{2}+\cdots\in
G_+.
\end{aligned}
\end{gather}
Also we define the symbols of $S^{-1},\bar S^{-1}$  as  $\bS^{-1},\bar \bS^{-1}$
\begin{gather}
\begin{aligned}
\bS^{-1}&=\I_N+\omega'_1(x)\lambda^{-1}+\omega'_2(x)\lambda^{-2}+\cdots,\\
\bar \bS^{-1}&=\bar\omega'_0(x)+\bar\omega'_1(x)\lambda+\bar\omega'_2(x)\lambda^{2}+\cdots.
\end{aligned}
\end{gather}

 The Lax  operators $\L,C_{kk},\bar C_{kk}\in\g$
 are defined by
\begin{align}
\label{Lax}  \L&:=W\cdot\Lambda\cdot W^{-1}=\bar W\cdot\Lambda^{-1}\cdot \bar W^{-1}, \\
\label{C} C_{kk}&:=W\cdot E_{kk}\cdot W^{-1},& \bar C_{kk}&:=\bar
W\cdot E_{kk}\cdot \bar W^{-1},
\end{align}
and
have the following expansions
\begin{gather}\label{lax expansion}
\begin{aligned}
 \L&=\Lambda+u_1(x)+u_2(x)\Lambda^{-1}, \\
C_{kk}&=E_{kk}+C_{kk,1}(x)\Lambda^{-1}+C_{kk,2}(x)\Lambda^{-2}+\cdots,\\
\bar \Cc_{kk}&=\bar C_{kk,0}(x)+\bar C_{kk,1}(x)\Lambda+\bar
C_{kk,2}(x)\Lambda^{2}+\cdots.
\end{aligned}
\end{gather}
 In fact the Lax  operators $\L,C_{kk},\bar C_{kk}\in\g$
 can also be equivalently defined by
\begin{align}
\label{Lax}  \L&:=S\cdot\Lambda\cdot S^{-1}=\bar S\cdot\Lambda^{-1}\cdot \bar S^{-1}, \\
\label{C} C_{kk}&:=S\cdot E_{kk}\cdot S^{-1},& \bar C_{kk}&:=\bar
S\cdot E_{kk}\cdot \bar S^{-1}.
\end{align}
These definitions are continuous interpolated versions of the multi-component Toda hierarchy, i.e. a continuous spatial parameter $x$ is brought into this hierarchy. Under this meaning, the continuous flow $\frac{\partial}{\partial x}$ is missing. To make these flows complete, we define the following logarithmic matrices

\begin{align}
\log_+\L&=(S\cdot\epsilon \partial\cdot S^{-1})= \epsilon \partial+\sum_{k < 0} W_k(x) \Lambda^k,\\
\log_-\L&=-(\bar S\cdot\epsilon \partial\cdot \bar S^{-1})=-\epsilon \partial+\sum_{k \geq 0} W_k(x) \Lambda^k,
\end{align}
where $\d$ is the derivative about spatial variable $x$.

Combining these above logarithmic operators together can help us in deriving the following important logarithmic matrix
\begin{align}
\label{Log} \log \L:&=\frac12\log_+\L+\frac12\log_-\L=\frac12(S\cdot\epsilon \partial\cdot S^{-1}-\bar S\cdot\epsilon \partial\cdot \bar S^{-1}),
\end{align}
which will generate a series of flow equations which contain the spatial flow in  Lax equations.

\section{ Lax equations of EMTH}

In this section we will use the factorization problem \eqref{facW} to derive  Lax equations.
Let us first introduce some convenient notations.
\begin{definition}The matrix operators $C_{kk},\bar C_{kk},B_{jk},\bar B_{jk},D_{j}$ are defined as follows
\begin{align}\label{satoS}
\begin{aligned}
C_{kk}&:=WE_{kk}W^{-1},\ \ \bar C_{kk}:=\bar W E_{kk}\bar W^{-1},\\
B_{jk}&:=WE_{kk}\Lambda^jW^{-1},\ \ \bar B_{jk}:=\bar W E_{kk}\Lambda^{-j}\bar W^{-1},\\
D_{j}&:=\frac{2\L^j}{j!}(\log \L-c_j),\ \ c_0=0;\ c_j=\sum_{i=1}^{j}\frac1i,j\geq 1.
\end{aligned}
\end{align}
\end{definition}

Now we give the definition of the extended multicomponent Toda hierarchy(EMTH).
\begin{definition}The extended multicomponent Toda hierarchy is a hierarchy in which the dressing operators $S,\bar S$ satisfy following Sato equations
\begin{align}
\label{satoSt} \epsilon\partial_{t_{jk}}S&=-(B_{jk})_-S,& \epsilon\partial_{t_{jk}}\bar S&=(B_{jk})_+\cdot\bar S,  \\
\label{satoSbart}
\epsilon\partial_{\bar t_{jk}}S&=-(\bar B_{jk})_-\cdot S,& \epsilon\partial_{\bar t_{jk}}\bar S&=(\bar B_{jk})_+\cdot\bar S, \\
\label{satoSs}\epsilon\partial_{ s_{j}}S&=-(D_{j})_-\cdot S,& \epsilon\partial_{s_{j}}\bar S&=(D_{j})_+\cdot\bar S.\end{align}
\end{definition}
Then one can easily get the following proposition about $W,\bar W.$

\begin{proposition}The wave operators $W,\bar W$ satisfy following Sato equations
\begin{align}
\label{Wjk} \epsilon\partial_{t_{jk}}W&=(B_{jk})_+\cdot W,& \epsilon\partial_{t_{jk}}\bar W&=(B_{jk})_+\cdot\bar W,  \\
\label{Wbjk}\epsilon\partial_{\bar t_{jk}}W&=-(\bar B_{jk})_-\cdot W,& \epsilon\partial_{\bar t_{jk}}\bar W&=-(\bar B_{jk})_-\cdot\bar W, \\
\epsilon\partial_{s_{j}}W&=\left(\frac{\L^j}{j!}(\log_+ \L-c_j) -(D_{j})_-\right)\cdot W,& \epsilon\partial_{s_{j}}\bar W&=\left(-\frac{\L^j}{j!}(\log_- \L-c_j)+(D_{j})_+\right)\cdot\bar W.  \end{align}
\end{proposition}

 From the previous proposition we can derive the following  Lax equations for the Lax operators.
\begin{proposition}\label{Lax}
 The  Lax equations of the EMTH are as follows
   \begin{align}
\label{laxtjk}
  \epsilon\partial_{t_{jk}} \L&= [(B_{jk})_+,\L],&
 \epsilon\partial_{t_{jk}} C_{ss}&= [(B_{jk})_+,C_{ss}],&\epsilon\partial_{t_{jk}} \bar C_{ss}&= [(B_{jk})_+,\bar C_{ss}],
\\
  \epsilon\partial_{\bar t_{jk}} \L&= [ (\bar B_{jk})_+,\L],&
 \epsilon\partial_{\bar t_{jk}} C_{ss}&= [(\bar B_{jk})_+,C_{ss}],&\epsilon\partial_{\bar t_{jk}} \bar C_{ss}&= [(\bar B_{jk})_+,\bar C_{ss}],
\\
 \epsilon\partial_{s_{j}} \L&= [(D_{j})_+,\L],&
  \epsilon\partial_{s_{j}} C_{ss}&= [(D_{j})_+,C_{ss}],& \epsilon\partial_{s_{j}} \bar C_{ss}&= [(D_{j})_+,\bar C_{ss}], \\ \label{logltjk}
  \epsilon\partial_{ t_{jk}} \log \L&= [(B_{jk})_+ ,\log \L],&\epsilon\partial_{\bar t_{jk}} \log \L&= [ -(\bar B_{jk})_-,\log \L],&&
  \end{align}
   \begin{align}\epsilon(\log \L)_{ s_{j}}=[ -(D_{j})_-, \frac{1}{2} \log_+ \L ]+
[(D_{j})_+ ,\frac{1}{2} \log_- \L ].
\end{align}
\end{proposition}

\begin{proof}
By eq.\eqref{Wjk} and  eq.\eqref{Wbjk}, one can
get eq.\eqref{logltjk} using dressing structures.

Similarly, the other Lax equations of EMTH can be derived using the double dressing structures and Sato equations eq.\eqref{satoSt}-eq.\eqref{satoSs}.
 \end{proof}

To see this kind of hierarchy more clearly, the $\d_{t_{1k}}$ flow equations  will be given in the next subsection.
\subsection{The multicomponent Toda equations}
 As a consequence of the factorization problem \eqref{facW} and  Sato equations, after taking into account that   $S\in G_-$ and $\bar S\in G_+$
and using  the notation $\Exp{\phi}:=\bar\omega_0$ in $\bar S$, $B_{1k}$ has following form
\begin{gather}\label{exp-omega}
\begin{aligned}
B_{1k}&=E_{kk}\Lambda+U_k+\bar U_k\Lambda^{-1},\ \ 1\leq k\leq N,
  \end{aligned}
\end{gather}
and we have the alternative expressions
\begin{gather}\label{exp-omega1}
\begin{aligned}
  U_k&:=\omega_1(x)E_{kk}-E_{kk}\omega_1(x+\epsilon)=\epsilon\partial_{t_{1k}}(\Exp{\phi(x)})\cdot\Exp{-\phi(x)},\\
 \bar U_k&= \Exp{\phi(x)}E_{kk}\Exp{-\phi(x-\epsilon)}=-\epsilon\partial_{t_{1k}}\omega_1(x).
\end{aligned}
\end{gather}

From Sato equations we deduce the following set of nonlinear
partial differential-difference equations
\begin{align}\left\{
\begin{aligned}
 \omega_1(x)E_{kk}-E_{kk}\omega_1(x+\epsilon)&=\epsilon\partial_{t_{1k}}(\Exp{\phi(x)})\cdot\Exp{-\phi(x)},\\
\epsilon\partial_{t_{1k}}\omega_1(x)&=-\Exp{\phi(x)}E_{kk}\Exp{-\phi(x-\epsilon)}.\end{aligned}\right.
\label{eq:multitoda}
\end{align}
These equations constitute what we call the multicomponent Toda equations. Observe that if we cross the two equations in \eqref{eq:multitoda}, then we get
\begin{align*}
  \epsilon\partial_{t_{1k}}\big(\epsilon\partial_{t_{1k}}(\Exp{\phi(x)})\cdot\Exp{-\phi(x)}\big)=
  E_{kk}\Exp{\phi(x+\epsilon)}E_{kk}\Exp{-\phi(x)}-\Exp{\phi(x)}E_{kk}\Exp{-\phi(x-\epsilon)}E_{kk},
\end{align*}
which is the matrix extension of the following Toda equation (the case when $N=1$)
\begin{align*}
  \epsilon^2\partial_{t_{11}}\partial_{t_{11}}(\phi(x))=\Exp{\phi(x+\epsilon)-\phi(x)}-
 \Exp{\phi(x)-\phi(x-\epsilon)}.
\end{align*}

Besides above multicomponent Toda equations, the logarithmic flows the EMTH also contains some extended flow equations in the next subsection.
\subsection{The extended equations}
Here we consider what the extended flow equations are like. Here we take the simplest case, i.e. the $s_{0}$ flow for $\L=\Lambda+u_0+u_1\Lambda^{-1},$
\[\epsilon\d_{s_{0}}\L&=&[(S\epsilon \d_x S^{-1})_+,\L]\\
&=&[\epsilon \d_xS S^{-1},\L]\\
&=&\epsilon\L_x,\]
which leads to  the following specific equation
\[\d_{s_{0}}u_0&=& u_{0x},\ \ \d_{s_{0}}u_1= u_{1x} .\]
The  above  flow equation tells us that $\d_{s_{0}}$ is just the $\d_x$ flow.

To give a linear description of the EMTH, we introduce matrix wave functions  $\psi,\bar\psi$ in the following part.
The matrix wave functions of the multi-component Toda hierarchy are
defined by
\begin{gather}\label{baker-fac}
\begin{aligned}
\psi&= W\cdot\chi, &
\bar\psi&=\bar W\cdot \bar\chi,
\end{aligned}
\end{gather}
where
\[
\chi(z):=z^{\frac{x}{\epsilon}}\mathbb I_N,\ \ \bar \chi(z):=z^{-\frac{x}{\epsilon}}\mathbb I_N.\
\]
Note that $\Lambda\chi=z\chi$ and  the following asymptotic expansions
are consequences of \eqref{expansion-S}
\begin{gather}\label{baker-asymp}
\begin{aligned}
  \psi&=z^{\frac{x}{\epsilon}}(\I_N+\omega_1(x)z^{-1}+\cdots)\,\psi_0(z),&\psi_0&:=\sum_{k=1}^NE_{kk}
 \Exp{\sum_{j=1}^\infty t_{jk}z^j+ s_{j}z^{j}\log z},& z&\rightarrow\infty,\\
\bar\psi&=z^{-\frac{x}{\epsilon}}(\bar\omega_0(x)+\bar\omega_1(x)z+\cdots)\,\bar\psi_0(z),
&\bar\psi_0&:=\sum_{k=1}^NE_{kk}
\Exp{\sum_{j=1}^\infty \bar t_{j k}z^{-j}+ s_{j}z^{-j}\log z},& z&\rightarrow 0.
\end{aligned}
\end{gather}

We can further get linear equations in the following proposition.

\begin{proposition}The matrix wave functions $\psi,\bar\psi$ are subject to following Sato equations
\begin{align}
 \L\psi&=z\psi,\ \ \ &&\L\bar\psi=z\bar\psi,\\
 \epsilon\partial_{jk}\psi&=(B_{jk})_+\cdot \psi,& \epsilon\partial_{jk}\bar \psi&=(B_{jk})_+\cdot\bar \psi,  \\
\epsilon\partial_{\bar t_{jk}}\psi&=-(\bar B_{jk})_-\cdot \psi,& \epsilon\partial_{\bar t_{jk}}\bar \psi&=-(\bar B_{jk})_-\cdot\bar \psi, \\
\epsilon\partial_{s_{j}}\psi&=(C_{ss}\L^j\log_+ \L -(D_{j})_-)\cdot \psi,& \epsilon\partial_{s_{j}}\bar \psi&=(-\bar C_{ss}\L^j\log_- \L+(D_{j})_+)\cdot\bar \psi.  \end{align}
\end{proposition}

\section{Hirota bilinear equations}
Basing on the above section, the Hirota bilinear equations which are equivalent to Lax equations of the EMTH can be derived in following proposition.
\begin{proposition}\label{HBEoper}
 $W$ and $\bar W$ are matrix wave operators of the multicomponent Toda hierarchy if and only the following Hirota bilinear equations hold true
\begin{align}\label{HBEW}
W\Lambda^r W^{-1}&=\bar W\Lambda^{-r}\bar W^{-1}, \ r\in \N.
   \end{align}
\end{proposition}

\begin{proof}
$\Rightarrow$ Firstly we will set

\begin{align}
\ga&=(\ga_{0,1},\ga_{1,1},\ga_{2,1},\ldots;\ga_{0,2},\ga_{1,2},\ga_{2,2},\ldots;\ldots;\ga_{0,N},\ga_{1,N},\ga_{2,N},\ldots;),\\
\bar \ga&=(\bar \ga_{0,1},\bar \ga_{1,1},\bar \ga_{2,1},\ldots;\bar \ga_{0,2},\bar \ga_{1,2},\bar \ga_{2,2},\ldots;\ldots;\bar \ga_{0,N},\bar \ga_{1,N},\bar \ga_{2,N},\ldots;),\\
 \beta&=(\beta_{1},\beta_{2},\ldots), \end{align}
be a multi index and
\begin{align}
\d^\ga:&=\d_{t_{0,1}}^{\ga_{0,1}}\d_{t_{1,1}}^{\ga_{1,1}}\d_{t_{2,1}}^{\ga_{2,1}}\ldots;
\d_{t_{0,2}}^{\ga_{0,2}}\d_{t_{1,2}}^{\ga_{1,2}}\d_{t_{2,2}}^{\ga_{2,2}}\ldots;\ldots;
\d_{t_{0,N}}^{\ga_{0,N}}\d_{t_{1,N}}^{\ga_{1,N}}\d_{t_{2,N}}^{\ga_{2,N}}\ldots\ ,\\
\d^{\bar \ga}:&=\d_{\bar t_{0,1}}^{\bar \ga_{0,1}}\d_{\bar t_{1,1}}^{\bar \ga_{1,1}}\d_{\bar t_{2,1}}^{\bar \ga_{2,1}}\ldots;
\d_{\bar t_{0,2}}^{\bar \ga_{0,2}}\d_{\bar t_{1,2}}^{\bar \ga_{1,2}}\d_{\bar t_{2,2}}^{\bar \ga_{2,2}}\ldots;\ldots;
\d_{\bar t_{0,N}}^{\bar \ga_{0,N}}\d_{\bar t_{1,N}}^{\bar \ga_{1,N}}\d_{\bar t_{2,N}}^{\bar \ga_{2,N}}\ldots\ ,\\
\d^\beta:&=\d_{s_{1}}^{\beta_{1}}\d_{s_{2}}^{\beta_{2}}\ldots .
\end{align}
Then we suppose $\d^{\theta}=\d^{\alpha}\d^{\bar\ga}\d^{\beta}$ ( we stress that $\d_{s_{0}}$ is not involved). We shall prove the left
statement leads to
\begin{eqnarray} \label{HBE2} W (x,t,\bar t,\Lambda)\Lambda^r
W^{-1}(x,t',\Lambda) = \bar W (x,t,\bar t,\Lambda)
\Lambda^{-r}\bar W^{-1}(x,t',\bar t',\Lambda)
\end{eqnarray}
 for all integers $r\geq 0$.
Using the same method used in\cite{M,ourJMP}, by induction on $\theta $,
we shall prove that
\begin{equation} \label{2.7}
W(x,t,\bar t,\Lambda)\Lambda^r(\d^\theta W^{-1}(x,t,\bar t,\Lambda))
=\bar W(x,t,\bar t,\Lambda) \Lambda^{-r}(\d^\theta\bar W^{-1}(x,t,\bar t,\Lambda)).
\end{equation} When $\theta=0$, it is obviously true according to the definition of
matrix wave operators.\\
Here firstly we suppose eq.\eqref{2.7} is true in the case of $\theta\neq 0$.
 Note that

\begin{equation}
\notag\epsilon\partial_{p_{jk}}
W :=
\begin{cases}
 [(\d_{t_{jk}}S)S^{-1}+SE_{kk}  \Lambda^jS^{-1}]W,
&p_{jk}=t_{jk},\\
 (\d_{\bar t_{jk}}S)S^{-1}W,& p_{jk}=\bar t_{jk},\\
[(\d_{s_{j}}S)S^{-1}+S  \Lambda^j\d_xS^{-1}]W,  &p_{jk}=s_{j},
  \end{cases}
\end{equation}
and
\begin{equation}
\notag\epsilon\partial_{p_{jk}}
\bar W :=
\begin{cases}
 (\d_{t_{jk}}\bar S)\bar S^{-1}\bar W,
&p_{jk}=t_{jk},\\
 [(\d_{\bar t_{jk}}\bar S)\bar S^{-1}+\bar SE_{kk}  \Lambda^{-j}\bar S^{-1}]\bar W,& p_{jk}=\bar t_{jk},\\
[(\d_{s_{j}}\bar S)\bar S^{-1}+\bar S \Lambda^{-j}\d_x\bar S^{-1}]\bar W,  &p_{jk}=s_{j},
  \end{cases}
\end{equation}
which further lead to \\
\begin{equation}
\notag\epsilon\partial_{p_{jk}}
W :=
\begin{cases}
 (B_{jk})_+W,
&p_{jk}=t_{jk},\\
 -(\bar B_{jk})_-W,& p_{jk}=\bar t_{jk},\\
[-(D_{j})_-+\frac{\L^j}{j!}(\log_+ \L-c_j)]W,  &p_{jk}=s_{j},
  \end{cases}
\end{equation}
and
\begin{equation}
\notag\epsilon\partial_{p_{jk}}
\bar W :=
\begin{cases}
 (B_{jk})_+\bar W,
&p_{jk}=t_{jk},\\
 -(\bar B_{jk})_-\bar W,& p_{jk}=\bar t_{jk},\\
[(D_{j})_+-\frac{\L^j}{j!}(\log_- \L-c_j)]\bar W,  &p_{jk}=s_{j}.
  \end{cases}
\end{equation}

This further implies
\begin{equation} \notag (\d_{p_{jk}}W)\Lambda^{r}(\d^\theta W^{-1}) =
(\d_{p_{jk}}\bar W)\Lambda^{-r}(\d^\theta\bar W^{-1}) \end{equation}
by considering \eqref{2.7} and furthermore we  get  \begin{equation} \notag
W\Lambda^r(\d_{p_{jk}}\d^\theta W^{-1} )=
\bar W\Lambda^{-r}(\d_{p_{jk}}\d^\theta\bar W^{-1}).
\end{equation} Thus if we increase the power of  $\d_{p_{jk}}$ by 1,
the eq.\eqref{2.7} still holds true.
 The induction is completed.
After doing Taylor expansion on both sides of eq.\eqref{HBE2} about $t=t',\bar t=\bar t',s=s'$, one can finish
the proof of eq.\eqref{HBE2}.

$\Leftarrow$ Vice versa, by separating the negative and the positive projection part about powers of shift operator $\Lambda$
in the equation \eqref{HBEW}, we can prove
$S, \ \bar S$ are a pair of matrix wave operators.
\end{proof}

To give a description in terms of matrix wave functions, following symbolic definitions are needed.
If the wave operator series have forms
\begin{eqnarray*} W(x,t,\bar t,s,\Lambda)=\sum_{i\in \Z} a_i(x,t,\bar t,s,
\d_x)\Lambda^i \mbox{ and } \bar W(x,t,\bar t,s,\Lambda)=\sum_{i\in \Z}
b_i(x,t,\bar t,s, \d_x)\Lambda^{i}, \end{eqnarray*}

\begin{eqnarray*}  W^{-1}(x,t,\bar t,s,\Lambda)=\sum_{i\in \Z}\Lambda^{i} a_i'(x,t,\bar t,s,
\d_x) \mbox{ and } \bar W^{-1}(x,t,\bar t,s,\Lambda)=\sum_{j\in
\Z}\Lambda^{j}b_j'(x,t,\bar t,s, \d_x),  \end{eqnarray*} then we denote their corresponding
left symbols $\W$,  $\bar \W$ and right symbols $\W^{-1}$, $\bar \W^{-1}$
as following
\begin{eqnarray*}
&&  \W(x,t,\bar t,s,\lambda) =\sum_{i\in \Z} a_i(x,t,\bar t,s,
\d_x)\lambda^i,\ \  \W^{-1}(x,t,\bar t,s,\lambda)=  \sum_{i\in \Z} a_i'(x,t,\bar t,s,
\d_x)\lambda^{i},\\
&&
\bar \W(x,t,\bar t,s,\lambda) =\sum_{i\in \Z} b_i(x,t,\bar t,s,
\d_x)\lambda^{i},\ \ \bar \W^{-1}(x,t,\bar t,s,\bar t,\lambda)=\sum_{j\in
\Z}b_j'(x,t,\bar t,s, \d_x)\lambda^{j}.
\end{eqnarray*}
With above preparation, it is time to give another form of Hirota bilinear equations (see the following proposition) after defining residue as $\res_{\lambda }\sum_{n\in \Z}\alpha_n \lambda^n=\alpha_{-1}$ using the similar proof as \cite{UT,M,ourJMP}.
\begin{proposition}\label{wave-operators}
Let  $s_{0} =s'_{0},$
 $S$ and
$\bar S$ are matrix-valued wave operators of the multicomponent Toda hierarchy if and only if for all   $m\in
\Z$, $r\in \N$, the following Hirota bilinear identity holds

\begin{eqnarray}  \notag &&\res_{\lambda }
 \left\{
\lambda^{r+m-1}\ \W(x,t,\bar t,s,\epsilon \partial_x,\lambda) \W^{-1}(x-m\epsilon,t',\bar t',s, \epsilon \partial_x,\lambda)
\right\} = \\ \label{HBE3}&& \res_{\lambda }
 \left\{
\lambda^{-r+m-1}\bar \W( x,t,\bar t,s,\epsilon \partial_x,\lambda )\
\bar \W^{-1}(x-m\epsilon,t',\bar t',s',\epsilon \partial_x,\lambda) \right\}.
\end{eqnarray}
\end{proposition}

\begin{proof}
 Let $m\in \Z$, $r\in \N$ and $s_{0} = s'_{0}$, then one can compare the
coefficients in front of $\Lambda^{-m}$ on both sides of eq.\eqref{HBE2} and find:
\begin{eqnarray*} \sum_{i+j=-m-r} a_i(x,t,\bar t,s, \d_x)a_j'(x-m\epsilon,t',\bar t',s',
\d_x) = \sum_{i+j=-m+r}  b_i(x,t,\bar t,s,\d_x)b_j'(x-m\epsilon,t',\bar t',s', \d_x).
\end{eqnarray*} This equality can be also rewritten  as eq.\eqref{HBE3}.

\end{proof}

To give Hirota quadratic functions in terms of tau functions, we need to define  tau functions and prove the existence of tau functions of the EMTH.

\section{Tau-functions of EMTH}
Firstly, we need to introduce the following sequences:
\[t-[\lambda] &:=& \left(t_{jk}-
  \frac{ \epsilon\lambda^j}{j}, 0\leq j\leq \infty,1\leq k\leq N\right), \ \ \\
  \bar t-[\lambda] &:=& \left(\bar t_{jk}-
  \frac{ \epsilon\lambda^j}{j}, 0\leq j\leq \infty,1\leq k\leq N\right).
\]
Referring to the Proposition 35 in \cite{manas}, the following tau function can be defined.
The matrix functions $\tau_M,\bar \tau_M$  depending only on the dynamical variables $t,\bar t$ and
$\epsilon$ are called the  {\em \bf Matrix tau-functions of the EMTH} if they
are related to the symbols of the matrix wave operators as following,
\begin{eqnarray}\label{pltaukk}(\bS)_{kk}: &=&\frac{ \tau
(s_{0,k}+x-\frac{\epsilon}{2}, t_{js}-\delta_{s,k}\frac{\epsilon}{j\lambda^j},\bar t,s;\epsilon) }
     {\tau (s_{0,k}+x-\frac{\epsilon}{2},t,\bar t,s;\epsilon)},\\ \label{pltausk}
     (\bS)_{sk}: &=&\lambda^{-1}\frac{ \tau_{sk}
(s_{0,k}+x-\frac{\epsilon}{2}, t_{jr}-\delta_{r,k}\frac{\epsilon}{j\lambda^j},\bar t,s;\epsilon) }
     {\tau (s_{0,k}+x-\frac{\epsilon}{2},t,\bar t,s;\epsilon)},\ s\neq k,\\
     \label{pl-1taukk}(\bS^{-1})_{kk}: &=&\frac{ \tau
(s_{0,k}+x+\frac{\epsilon}{2}, t_{js}+\delta_{s,k}\frac{\epsilon}{j\lambda^j},\bar t,s;\epsilon) }
     {\tau (s_{0,k}+x+\frac{\epsilon}{2},t,\bar t,s;\epsilon)},\\ \label{pl-1tausk}
     (\bS^{-1})_{sk}: &=&\lambda\frac{ \tau_{sk}
(s_{0,k}+x+\frac{\epsilon}{2}, t_{jr}+\delta_{r,k}\frac{\epsilon}{j\lambda^j},\bar t,s;\epsilon) }
     {\tau (s_{0,k}+x+\frac{\epsilon}{2},t,\bar t,s;\epsilon)},\ s\neq k,\\\label{prtaukk}
(\bar \bS)_{kk}:&= &\frac{ \tau_{kk}
(s_{0,k}+x+\frac{\epsilon}{2},t,\bar t_{js}+\delta_{s,k}\frac{\epsilon\lambda^j}{j},s;\epsilon)}
     {\tau(s_{0,k}+x-\frac{\epsilon}{2},t,\bar t,s;\epsilon)},\\\label{prtausk}
     (\bar \bS)_{sk}:&= &\frac{ \bar \tau_{sk}
(s_{0,k}+x+\frac{3\epsilon}{2},t,\bar t_{js}+\delta_{s,k}\frac{\epsilon\lambda^j}{j},s;\epsilon)}
     {\tau(s_{0,k}+x+\frac{\epsilon}{2},t,\bar t,s;\epsilon)},\ s\neq k,\\
     (\bar \bS^{-1})_{kk}:&= &\frac{\tau_{kk}
(s_{0,k}+x-\frac{\epsilon}{2},t,\bar t_{js}-\delta_{s,k}\frac{\epsilon\lambda^j}{j},s;\epsilon)}
     {\tau(s_{0,k}+x+\frac{\epsilon}{2},t,\bar t,s;\epsilon)},\\\label{prtausk}
     (\bar \bS^{-1})_{sk}:&= &\frac{\bar  \tau_{sk}
(s_{0,k}+x+\frac{\epsilon}{2},t,\bar t_{js}-\delta_{s,k}\frac{\epsilon\lambda^j}{j},s;\epsilon)}
     {\tau(s_{0,k}+x+\frac{3\epsilon}{2},t,\bar t,s;\epsilon)},\ s\neq k.
     \end{eqnarray}
For convenience, we denote two matrices $\tau_M$ and $\bar \tau_M$ as
\begin{eqnarray} \label{taum}
(\tau_M)_{ij}=
\begin{cases}
\vspace{.1in}
\tau_{ii}=\tau,\ \ i=j\\
\vspace{.1in}
\tau_{ij}, \ \ i\neq j,
\end{cases}
\ \ \
(\bar \tau_M)_{ij}=
\begin{cases}
\vspace{.1in}
\bar \tau_{ii},\ \ i=j\\
\vspace{.1in}
\bar \tau_{ij}, \ \ i\neq j.
\end{cases}
\end{eqnarray}
Then we can rewrite the definition of tau functions as
\begin{eqnarray}\label{Mpltaukk}\bS: &=&\frac{ \tau_M
(s_{0,k}+x-\frac{\epsilon}{2}, t_{js}-\delta_{s,k}\frac{\epsilon}{j\lambda^j},\bar t,s;\epsilon) }
     {\tau (s_{0,k}+x-\frac{\epsilon}{2},t,\bar t,s;\epsilon)},\\
     \label{Mpl-1taukk}\bS^{-1}: &=&\frac{ \tau_M
(s_{0,k}+x+\frac{\epsilon}{2}, t_{js}+\delta_{s,k}\frac{\epsilon}{j\lambda^j},\bar t,s;\epsilon) }
     {\tau (s_{0,k}+x+\frac{\epsilon}{2},t,\bar t,s;\epsilon)},\\ \label{Mprtaukk}
\bar \bS:&= &\frac{ \bar \tau_M
(s_{0,k}+x+\frac{\epsilon}{2},t,\bar t_{js}+\delta_{s,k}\frac{\epsilon\lambda^j}{j},s;\epsilon)}
     {\tau(s_{0,k}+x-\frac{\epsilon}{2},t,\bar t,s;\epsilon)},\\
     \bar \bS^{-1}:&= &\frac{\bar \tau_M
(s_{0,k}+x-\frac{\epsilon}{2},t,\bar t_{js}-\delta_{s,k}\frac{\epsilon\lambda^j}{j},s;\epsilon)}
     {\tau(s_{0,k}+x+\frac{\epsilon}{2},t,\bar t,s;\epsilon)}.
     \end{eqnarray}
To give the existence of tau function, we need the following two lemmas.
\begin{lemma}The following identities about $\bS(x,t,\lambda),\bS^{-1}(x,t,\lambda),\bar \bS(x,t,\lambda),\bar \bS^{-1}(x,t,\lambda)$ hold
\begin{eqnarray}
&&\bS(x,\t,\lambda_{1})_{kk}\bS^{-1}(x,t-[\lambda_1^{-1}]_k,\bar t+[\lambda_{2}]_k,\lambda_{1})_{kk}\\ \notag
\label{j1}
&&= \bar \bS(x,\t,\lambda_{2})_{kk} \bar \bS^{-1}(x,t-[\lambda_1^{-1}]_k,\bar t+[\lambda_{2}]_k,\lambda_{2})_{kk}+\sum_{s\neq k}\tilde \omega_0(x,\t)_{ks}\tilde \omega_0^{-1}(x,t-[\lambda_1^{-1}]_{k},\bar t+[\lambda_2]_{k})_{sk},\\
\label{j2} &&
\bS(x,\t,\lambda)_{kk}\bS^{-1}(x-\epsilon,t-[\lambda^{-1}]_k,\bar t,\lambda)_{kk} = 1,\\
\notag &&
\bS(x,\t,\lambda_1)_{kk}\bS^{-1}(x-\epsilon,t-[\lambda_1^{-1}]_k-[\lambda_2^{-1}]_k,\bar t,\lambda_1)_{kk}\\
\label{j3}&&
=\bS(x,\t,\lambda_2)_{kk}\bS^{-1}(x-\epsilon,t-[\lambda_1^{-1}]_k-[\lambda_2^{-1}]_k,\lambda_2)_{kk},\\
&&
 \bar \bS(x,\t,\lambda_1)_{kk} \bar \bS^{-1}(x+\epsilon,t,\bar t+[\lambda_1]_k+[\lambda_2]_k,\lambda_1)_{kk}\\ \notag
 \label{j4} &&
= \bar \bS(x,\t,\lambda_2)_{kk} \bar \bS^{-1}(x+\epsilon,t,\bar t+[\lambda_1]_k+[\lambda_2]_k,\lambda_2)_{kk}+(\lambda_2-\lambda_1)\tilde \omega'_1(x+\epsilon,\bar t+[\lambda_1]_{k}+[\lambda_2]_{k})_{kk},  \\
\label{i5} &&  \bar \bS(x,\t,\lambda)_{kk} \bar \bS^{-1}(x,t,\bar t+[\lambda]_k,\lambda)_{kk}+\sum_{s\neq k}\tilde \omega_0(x,\t)_{ks}\tilde \omega_0^{-1}(x,t,\bar t+[\lambda]_{k})_{sk} =1.
\end{eqnarray}
\end{lemma}

In the following part, we sometimes denote $\t$ as $(t,\bar t)$ for short.
\begin{lemma}The following identities about $\bS(x,t,\lambda),\bS^{-1}(x,t,\lambda),\bar \bS(x,t,\lambda),\bar \bS^{-1}(x,t,\lambda)$ hold
\label{identities'}
\begin{eqnarray} \label{lr}
&&\bS(x,t,\lambda_{1})_{ik}\bS^{-1}(x,t-[\lambda_1^{-1}]_k,\bar t+[\lambda_{2}]_k,\lambda_{1})_{kj}+\delta_{ij}-\delta_{ik}\delta_{kj}\\
\notag
&&= \bar \bS(x,t,\lambda_{2})_{ik} \bar \bS^{-1}(x,t-[\lambda_1^{-1}]_k,\bar t+[\lambda_{2}]_k,\lambda_{2})_{kj}+\sum_{s\neq k}\tilde \omega_0(x,\t)_{is}\tilde \omega_0^{-1}(x,t-[\lambda_1^{-1}]_{k},\bar t+[\lambda_2]_{k})_{sj},\\
\notag &&
\omega_1(x,\t)_{ij}+\delta_{ik}\omega_1(x,t-[\lambda_1^{-1}]_k-[\lambda_2^{-1}]_k)_{kj}+\bS(x,\t,\lambda_1)_{ik}\bS^{-1}(x-\epsilon,t-[\lambda_1^{-1}]_k-[\lambda_2^{-1}]_k,\bar t,\lambda_1)_{kj}\\ \notag
\label{i2}&&
=\delta_{jk}\omega_1(x,\t)_{ik}+\omega_1(x,t-[\lambda_1^{-1}]_k-[\lambda_2^{-1}]_k)_{kj}+\bS(x,\t,\lambda_2)_{ik}\bS^{-1}(x-\epsilon,t-[\lambda_1^{-1}]_k-[\lambda_2^{-1}]_k,\lambda_2)_{kj},\\
\notag
 &&\bar \bS(x,\t,\lambda_1)_{rk}\bar \bS^{-1}(x+\epsilon,\bar t+[\lambda_1]_{k}+[\lambda_2]_{k},\lambda_1)_{kk}-\bar \bS(x,\t,\lambda_2)_{rk}\bar \bS^{-1}(x+\epsilon,\bar t+[\lambda_1]_{k}+[\lambda_2]_{k},\lambda_2)_{kk}\\ \notag
&&=(\lambda_2-\lambda_1)[\sum_{s\neq k}\tilde \omega_1(x,t)_{rs}\tilde \omega_0^{-1}(x+\epsilon,\bar t+[\lambda_1]_{k}+[\lambda_2]_{k})_{sk}-\sum_{s\neq k}\tilde \omega_0(x,t)_{rs}\tilde \omega'_1(x,\bar t+[\lambda_1]_{k}+[\lambda_2]_{k})_{sk}],  \\
\label{i4'} &&  \bar \bS(x,\t,\lambda)_{ik} \bar \bS^{-1}(x,t,\bar t+[\lambda]_k,\lambda)_{kj}+\sum_{s\neq k}\tilde \omega_0(x,\t)_{is}\tilde \omega_0^{-1}(x,t,\bar t+[\lambda]_{k})_{sj} =
\delta_{ij},\\ \label{plx}
&&\bS(x,\t,\lambda_1)_{ss}\bS^{-1}(x,\t-[\lambda_1^{-1}]_{s},\lambda_1)_{ss}=\sum_{k}\tilde \omega_0(x,\t)_{sk}\tilde \omega_0^{-1}(x,t-[\lambda_1^{-1}]_{k},\bar t)_{ks}.\end{eqnarray}
\end{lemma}

\begin{proof} Let $\t$ and $\t'$ be two sequences of time variables such
that $s_{n}=s'_{n}$, $n\geq 0$. The identities
$\eqref{lr}-\eqref{i4'}$ are consequence of the following one:
\begin{eqnarray} \notag
 && S(x,\t,\Lambda)\sum_{k=1}^NE_{kk}
 \exp(\sum_{j=1}^\infty (t_{jk}-t'_{jk})\Lambda^j)
S^{-1}(x,\t',\Lambda)\\
\label{chan}
& =& \ \bar S(x,\t,\Lambda)\sum_{k=1}^NE_{kk}\exp(-\sum_{j=1}^\infty (\bar t_{jk}-\bar t'_{jk})\Lambda^{-j})
\ \bar S^{-1}(x,\t',\Lambda).
 \end{eqnarray} \\
 The proof of
\eqref{chan} is completely analogous to the argument in the
implication from left statement to the right result in Proposition \ref{HBEoper}
and it will be omitted.
 To prove eq.\eqref{lr}, in
eq.\eqref{chan} we put $t'_{jk}=t_{jk}-[\lambda_1^{-1}]_k+[\lambda_2]_k$.
  The exponential factor turns into
\begin{eqnarray*} \exp \Big(\sum_{j=0}^\infty \Lambda^j\lambda^{-j}\Big) =
 (1-\lambda^{-1}\Lambda)^{-1}=\sum_{s\geq 0}(\lambda^{-1}\Lambda)^s,
\end{eqnarray*}

\begin{eqnarray*} \exp \left(\sum_{j=0}^\infty \lambda^j\Lambda^{-j}
\right) =
 (1-\lambda\Lambda^{-1})^{-1}=\sum_{s\geq 0}(\lambda\Lambda^{-1})^s.
\end{eqnarray*}

The other identities can be proved similarly as \cite{M}.
\end{proof}

The above two lemmas can be rewritten into the following single lemma.
\begin{lemma}The following equations hold
\label{identities2s'}
\begin{eqnarray} \label{lr2}
&&\sum_{k=1}^{N}\bS(x,\t,\lambda_{1k})_{ik}\bS^{-1}(x+\epsilon,t,\bar t+[\lambda_{2k}]_k,\lambda_{1k})_{kj}= \sum_{k=1}^{N}\bar \bS(x,\t,\lambda_{2k})_{ik}\bar \bS^{-1}(x,t-[\lambda_{1k}^{-1}]_k,\bar t,\lambda_{2k})_{kj},\notag\\ \\
\label{i22}&&
\sum_{k=1}^{N}\bS(x,\t,\lambda_{1k})_{ik}\bS^{-1}(x,t-[\lambda_{2k}^{-1}]_k,\bar t,\lambda_{1k})_{kj}
=\sum_{k=1}^{N}\bS(x,\t,\lambda_{2k})_{ik}\bS^{-1}(x,t-[\lambda_{1k}^{-1}]_k,\bar t,\lambda_{2k})_{kj},\\
\notag&& \notag
 \sum_{k=1}^{N}\bar \bS(x,\t,\lambda_{1k})_{ik}\bar \bS^{-1}(x+\epsilon,t,\bar t+[\lambda_{2k}]_k,\lambda_{1k})_{kj}
=\sum_{k=1}^{N}\bar \bS(x,\t,\lambda_{2k})_{ik}\bar \bS^{-1}(x+\epsilon,t,\bar t+[\lambda_{1k}]_k,\lambda_{2k})_{kj}.\\
\end{eqnarray}
\end{lemma}

Using these lemmas above, one can prove the existence of tau function in the following proposition.

\begin{proposition}\label{tau-function}
Given a pair of wave operators $\bS$ and $\ \bar \bS$ of the METH there
exists corresponding tau-functions $\tau_{sk},\bar \tau_{sk}$, which is unique up to
multiplication by a non-vanishing function independent of $t_{sk}$, $\bar t_{sk}$.
\end{proposition}
\begin{proof}
The proof is quite complicated but standard which will not be mentioned here. One can check the references \cite{M,ourJMP,manas}.
\end{proof}

With the above preparation, we will give the Hirota bilinear equations in terms of tau functions in the next
section with the help of generalized vertex operators.

\section{Generalized matrix vertex operators and  Hirota quadratic equations}
In this section we continue to  discuss on the fundamental properties
of the tau function of the EMTH, i.e., the Hirota quadratic equations of the EMTH. So we
introduce the following vertex operators
\begin{eqnarray*}
\Gamma^{\pm a} :&=&\exp\left(\pm \frac{1}{\epsilon} \sum_{k=1}^NE_{kk}
(\sum_{j=0}^\infty t_{jk}\lambda^j+ s_{j}\lambda^j log \lambda)\right)\times\exp\left({\mp
\frac{\epsilon}{2}\partial_{s_{0}} \mp [\lambda^{-1}]_\d  }\right),\\
\Gamma^{\pm b} :&=&\exp\left(\pm  \frac{1}{\epsilon}\sum_{k=1}^NE_{kk}
(\sum_{j=0}^\infty \bar t_{jk}\lambda^{-j}- s_{j}\lambda^{-j} log \lambda)\right)\times\exp\left({\mp
\frac{\epsilon}{2}\partial_{s_{0}} \mp [\lambda]_{\bar \d}  }\right).
\end{eqnarray*}

where\ \
\begin{eqnarray*}
 [\lambda]_\d  :&=&
\epsilon\sum_{k=1}^N\sum_{j=0}^\infty \frac{\lambda^{j}}{j}\partial_{t_{jk}},\ \ [\lambda]_{\bar \d}  :=\epsilon\sum_{k=1}^N
\sum_{j=0}^\infty \frac{\lambda^j}{j}\partial_{\bar t_{jk}}.
\end{eqnarray*}

Because of the logarithm $\log \lambda$, the vertex
operators  $\Gamma^{\pm a} \otimes \Gamma^{\mp a}$ and
$\Gamma^{\pm b} \otimes \Gamma^{\mp b }$  are multi-valued
function. There
are monodromy factors $M^a$ and $M^b$ respectively as following
among different branches around $\lambda=\infty$
\begin{equation} M^{a}= \exp \left\{ \pm \frac{2\pi i}{\epsilon} \sum_{k=1}^NE_{kk}
\sum_{j\geq 0}\lambda^j ( s_{j} \otimes 1 - 1\otimes s_{j})
\right\},\end{equation}
\begin{equation}
 M^{b}= \exp \left\{ \pm \frac{2\pi i}{\epsilon} \sum_{k=1}^NE_{kk}
\sum_{j\geq 0}\lambda^{-j} ( s_{j} \otimes 1 - 1\otimes s_{j})
\right\}.
\end{equation}
In order to offset the complication, we need to generalize the
concept of vertex operators which leads it to be not scalar-valued
any more but take values in a differential operator algebra. So we introduce the following vertex operators
\begin{equation}\Gamma^{\delta}_{a} = \exp\( -\sum_{k=1}^NE_{kk}\sum_{j>0}\frac{\lambda^{j}}{\epsilon
j}(\epsilon\d_x)s_{j}\) \exp(x\partial_{s_{0}}),\end{equation}
\begin{equation}\Gamma^{\delta}_{b} = \exp\( -\sum_{k=1}^NE_{kk}\sum_{j>0}\frac{\lambda^{-j}}{\epsilon
j}(\epsilon\d_x)s_{j}\) \exp(x\partial_{s_{0}}),\end{equation}
\begin{equation}\Gamma^{\delta
\#}_{a} =\exp(x\partial_{s_{0}}) \exp\( \sum_{k=1}^NE_{kk}\sum_{j>0}\frac{\lambda^{j}}{\epsilon
j}(\epsilon\d_x)s_{j}\) ,\end{equation}
\begin{equation}\Gamma^{\delta\#}_{b} =  \exp(x\partial_{s_{0}})\exp\( \sum_{k=1}^NE_{kk}\sum_{j>0}\frac{\lambda^{-j}}{\epsilon
j}(\epsilon\d_x)s_{j}\).\end{equation}
 Then \begin{equation}
 \label{double delta a} \Gamma^{\delta
\#}_{a}\otimes \Gamma^{\delta}_{a} = \exp(x\partial_{s_{0}})\exp\(
\sum_{j>0}\frac{\lambda^{j}}{\epsilon
j}(\epsilon\d_x)(s_{j}-s'_{j}) \) \exp(x\partial_{s'_{0}}),
\end{equation}
\begin{equation}
\label{double delta b} \ \ \ \Gamma^{\delta \#}_{b}\otimes
\Gamma^{\delta}_{b} = \exp(x\partial_{s_{0}})\exp\(
\sum_{j>0}\frac{\lambda^{-j}}{\epsilon
j}(\epsilon\d_x)(s_{j}-s'_{j}) \) \exp(x\partial_{s'_{0}}).
\end{equation}

After computation we get
\begin{eqnarray*} && \(\Gamma^{\delta \#}_{a}\otimes \Gamma^\delta_{a} \) M^{a} =
\exp \left\{ \pm \frac{2\pi i}{\epsilon}\sum_{k=1}^NE_{kk} \sum_{j> 0} \frac{\lambda^{j}}{j} ( s_{j}-s'_{j})
\right\}\\
&& \exp\(  \pm \frac{2\pi i}{\epsilon} ((s_{0}+x) -(s'_{0}+x+ \sum_{j> 0} \frac{\lambda^{j}}{j} ( s_{j}-s'_{j})) \) \(\Gamma^{\delta \#}_{a}\otimes \Gamma^\delta _{a}\)
\\&=& \exp\({\pm \frac{2\pi i}{\epsilon}(s_{0}-s'_{0})}\)
\(\Gamma^{\delta \#}_{a}\otimes \Gamma^\delta _{a}\),
\end{eqnarray*}
\begin{eqnarray*}
&& \(\Gamma^{\delta \#}_{b}\otimes \Gamma^\delta_{ b} \) M^{b} =
\exp \left\{ \pm \frac{2\pi i}{\epsilon} \sum_{k=1}^NE_{kk}\sum_{j> 0} \frac{\lambda^{-j}}{j} ( s_{j}-s'_{j})
\right\}\\
&& \exp\(  \pm \frac{2\pi i}{\epsilon} ( (s_{0}+x) -\sum_{k=1}^NE_{kk}(
s'_{0}+x+ \sum_{j> 0} \frac{\lambda^{-j}}{j} ( s_{j}-s'_{j})) \)\(\Gamma^{\delta \#}_{b}\otimes \Gamma^\delta_{b} \)
\\&=& \exp\({\pm \frac{2\pi i}{\epsilon}(s_{0}-s'_{0})}\)
\(\Gamma^{\delta \#}_{b}\otimes \Gamma^\delta_{b} \).
\end{eqnarray*}
Thus when $s_{0}-s'_{0} \in \Z\epsilon $, $\(\Gamma^{\delta
\#}_{a}\otimes \Gamma^{\delta}_{a}\) \( \Gamma^{a}\otimes
\Gamma^{-a}\) \mbox{and}\(\Gamma^{\delta \#}_{b}\otimes
\Gamma^{\delta}_{b }\)\(\Gamma^{-b}\otimes\Gamma^{b}\)$ are all
single-valued near $\lambda=\infty$.

 Now we should note that the above vertex operators  take
 value in differential operator algebra $\C[\d,x,t,\bar t,s,\epsilon]:=\{f(x,t,\epsilon)|f(x,t,\bar t,s,\epsilon)=\sum_{i\geq 0}c_{i}(x,t,\bar t,s,\epsilon)\d^i\}$.

\begin{theorem}\label{t11}
Function $\tau_M,\bar \tau_M$  are tau-functions of the EMTH if and only if they
satisfy the  following Hirota quadratic equations  of the EMTH
\begin{equation} \label{HBE} \res_{{\rm{\lambda}}}
 \left(\lambda^{r-1}\(\Gamma^{\delta
\#}_{a}\otimes \Gamma^{\delta}_{a}\) \( \Gamma^{a}\otimes
\Gamma^{-a}\right)(\tau_M
\otimes \tau_M ) -\lambda^{-r-1}\(\Gamma^{\delta \#}_{b}\otimes
\Gamma^{\delta}_{b }\)\(\Gamma^{-b}\otimes\Gamma^{b} \) (\bar \tau_M
\otimes \bar\tau_M )\right)=0
\end{equation}
computed at $s_{0}-s'_{0}=l\epsilon$
 for each  $l\in \Z$, $r\in \N$.
\end{theorem}
 \begin{proof}

 We just need  to prove that the HBEs
are equivalent to the right side in Proposition
\ref{wave-operators}. By a straightforward computation we can get
the following four identities {\allowdisplaybreaks}
\begin{eqnarray}\label{vertex computation1}
\Gamma^{\delta \#}_a \Gamma^{a}\tau_M & =& \tau(s_{0}+x-\frac{\epsilon}{2},t,\bar t,s)
\lambda^{\sum_{k=1}^NE_{kk}s_{0}/\epsilon} \W(x,t,\bar t,s,\epsilon \d_x,\lambda )\lambda^{\I_Nx/\epsilon},
\\ \label{vertex computation2}
 \Gamma^{\delta}_a \Gamma^{-a}\tau_M  & =&
\lambda ^{-(s_{0}+x)/\epsilon}
\W^{-1}(x,t,\bar t,s,\epsilon\d_x,\lambda )\tau(x+s_{0}+\frac{\epsilon}{2},t,\bar{t},s), \\\label{vertex
computation3}
 \Gamma^{\delta \#}_b \Gamma^{-b}\bar \tau_M  & =&
\tau(x+s_{0}-\frac{\epsilon}{2},t,\bar t,s) \lambda^{\sum_{k=1}^NE_{kk}s_{0}/\epsilon} \bar \W(x,t,\bar t,s,\epsilon \d_x,\lambda
)\lambda^{ x\I_N/\epsilon}, \\ \notag
 \Gamma^{\delta}_b \Gamma^{b}\bar\tau_M & = &\lambda^{-\sum_{k=1}^NE_{kk}s_{0}/\epsilon} \lambda^{
 -x\I_N/\epsilon} \bar \W^{-1}(x,t,\bar t,s,\epsilon \d_x,\lambda)\ \label{vertex computation4}
\tau(x+s_{0}+\frac{\epsilon}{2},t,\bar t,s) .\\
\end{eqnarray}

As an example, we only give the proof for the eq.\eqref{vertex computation2},

 \begin{eqnarray*}
&&\Gamma^{\delta }_a \Gamma^{-a}\tau_M\\
 &=&\exp\( -\sum_{k=1}^NE_{kk}\sum_{j>0}\frac{\lambda^{j}}{\epsilon
j}(\epsilon\d_x)s_{j}\) \exp(x\partial_{s_{0}})\\
&&\exp\left(-\frac{1}{\epsilon} \sum_{k=1}^NE_{kk}
(\sum_{j=0}^\infty t_{jk}\lambda^j+ s_{j}\lambda^j log \lambda)\right)\times\exp\left({
\frac{\epsilon}{2}\partial_{s_{0}} + [\lambda^{-1}]_\d  }\right)\tau_M\\
&=& \exp \left\{\left[-\frac{1}{\epsilon} \sum_{k=1}^NE_{kk}
\sum_{j=0}^\infty t_{jk}\lambda^j
-\sum_{k=1}^NE_{kk}\sum_{j>0}\frac{\lambda^{j}}{\epsilon
j}(\epsilon\d_x)s_{j}\right] \right\}\\
&& \exp\(-\frac{1}{\epsilon}\sum_{k=1}^NE_{kk}
\sum_{j>0}^\infty s_{j}\lambda^j log \lambda\)\exp\left\{-(\log \lambda
)(s_{0}+x)/\epsilon
 \right\}\\
 &&\tau_M(x+s_{0}+\frac{\epsilon}{2},t+[\lambda^{-1}],\bar{t},s)\\
& =&\lambda ^{-(s_{0}+x)/\epsilon} \exp \left\{\left[-\frac{1}{\epsilon} \sum_{k=1}^NE_{kk}
\sum_{j=0}^\infty t_{jk}\lambda^j
-\sum_{k=1}^NE_{kk}\sum_{j>0}\frac{\lambda^{j}}{\epsilon
j}(\epsilon\d_x)s_{j}\right] \right\}\\
&&\tau_M(x+s_{0}+\frac{\epsilon}{2},t+[\lambda^{-1}],\bar{t},s)\\
&=&\lambda ^{-(s_{0}+x)/\epsilon} \exp \left\{\left[-\frac{1}{\epsilon} \sum_{k=1}^NE_{kk}
\sum_{j=0}^\infty t_{jk}\lambda^j
-\sum_{k=1}^NE_{kk}\sum_{j>0}\frac{\lambda^{j}}{\epsilon
j}(\epsilon\d_x)s_{j}\right] \right\}
\bS^{-1}(x,t,\bar t,\lambda)\\
&&\tau(x+s_{0}+\frac{\epsilon}{2},t,\bar{t},s)\\
&=&
 \lambda ^{-(s_{0}+x)/\epsilon}
\W^{-1}(x,t,\epsilon\d_x,\lambda )\tau(x+s_{0}+\frac{\epsilon}{2},t,\bar{t},s).
\end{eqnarray*}
So eq.\eqref{vertex computation2} is proved. Eq.\eqref{vertex
computation1},\eqref{vertex
computation3},eq.\eqref{vertex computation4} can be proved in similar ways.
By substituting four equations eq.\eqref{vertex
computation1}-eq.\eqref{vertex computation4} into the HBEs
\eqref{HBE},
eq.\eqref{HBE3} is derived.
\end{proof}

The eq.\eqref{HBE} in the case when $N=1$ is exactly the
  Hirota quadratic equation of the extended Toda hierarchy in \cite{M}. To give more information on the relations among different
solutions of the EMTH, the Darboux transformation of the EMTH will be constructed using kernel determinant technique as \cite{Hedeterminant,rogueHMB} in the next section.

\section{Darboux transformations of the EMTH}

In this section, we will consider the Darboux transformations of the EMTH on the Lax operator
 \[\L=\Lambda+u+v\Lambda^{-1},\]
 i.e.
  \[\label{1darbouxL}\L^{[1]}=\Lambda+u^{[1]}+v^{[1]}\Lambda^{-1}=W\L W^{-1},\]
where $W$ is the Darboux transformation operator.

That means after the Darboux transformation, the spectral problem about $N\times N$ spectral matrix $\phi$

\[\L\phi=\Lambda\phi+u\phi+v\Lambda^{-1}\phi=\lambda\phi,\]
will become

\[\L^{[1]}\phi^{[1]}=\lambda\phi^{[1]}.\]

To keep the Lax pair of the EMTH invariant, i.e.
   \begin{align}
\label{laxtjk}
  \epsilon \partial_{t_{jk}} \L^{[1]}&= [(B_{jk}^{[1]})_+,\L^{[1]}],&
 \epsilon \partial_{t_{jk}} C_{ss}^{[1]}&= [(B_{jk}^{[1]})_+,C_{ss}^{[1]}],&\epsilon \partial_{t_{jk}} \bar C_{ss}^{[1]}&= [(B_{jk}^{[1]})_+,\bar C_{ss}^{[1]}],
\\
  \epsilon \partial_{\bar t_{jk}} \L^{[1]}&= [ (\bar B_{jk}^{[1]})_+,\L^{[1]}],&
 \epsilon \partial_{\bar t_{jk}} C_{ss}^{[1]}&= [(\bar B_{jk}^{[1]})_+,C_{ss}^{[1]}],&\epsilon \partial_{\bar t_{jk}} \bar C_{ss}^{[1]}&= [(\bar B_{jk}^{[1]})_+,\bar C_{ss}^{[1]}],
\\
 \epsilon \partial_{s_{j}} \L^{[1]}&= [(D_{j}^{[1]})_+,\L^{[1]}],&
  \epsilon \partial_{s_{j}} C_{ss}^{[1]}&= [(D_{j}^{[1]})_+,C_{ss}^{[1]}],& \epsilon \partial_{s_{j}} \bar C_{ss}^{[1]}&= [(D_{j}^{[1]})_+,\bar C_{ss}^{[1]}], \\ \label{logltjk'}
  \epsilon \partial_{ t_{jk}} \log \L^{[1]}&= [(B_{jk}^{[1]})_+ ,\log \L^{[1]}],&\epsilon \partial_{\bar t_{jk}} \log \L^{[1]}&= [ -(\bar B_{jk}^{[1]})_-,\log \L^{[1]}],&&
  \end{align}
   \begin{align}\epsilon(\log \L^{[1]})_{ s_{j}}=[ -(D_{j}^{[1]})_-, \frac{1}{2} \log_+ \L^{[1]} ]+
[(D_{j}^{[1]})_+ ,\frac{1}{2} \log_- \L^{[1]} ],
\end{align}
\begin{equation}
 B_{\alpha,n}^{[1]}:=B_{\alpha,n}(\L^{[1]}), \ \ \ \bar B_{\alpha,n}^{[1]}:=\bar B_{\alpha,n}(\L^{[1]}),
\end{equation} the dressing operator $W$ should satisfy the following dressing equations
\[\epsilon W_{t_{j,n}}&=&-W(B_{j,n})_++(WB_{j,n}W^{-1})_+W,\ \ 1 \leq n\leq N, j\geq 0,\\
\epsilon W_{\bar t_{j,n}}&=&-W(B_{j,n})_++(WB_{j,n}W^{-1})_+W,\ \ 1 \leq n\leq N, j\geq 0,\\
\epsilon W_{s_{j}}&=&-W(D_{j})_++(WD_{j}W^{-1})_+W,\ \ j\geq 0.\]
where $W_{t_{\gamma,n}}$ means the derivative of $W$ by $t_{\gamma,n}.$ For a local operator $B=\sum_{m=0}^{\infty}b_m(x) \Lambda^{m}$, we define
$B^*(g(x))=\sum_{m=0}^{\infty} \Lambda^{-m}(g(x) b_m(x)).$
To give the Darboux transformation, we need the following lemma.

\begin{lemma}\label{lema}
The operator $B:=\sum_{n=0}^{\infty}b_n\La^n$ is a non-negative matrix-valued difference operator,  $C:=\sum_{n=1}^{\infty}c_n\La^{-n}$ is a negative  matrix-valued difference operator and $f,g$ (short for $f(x),g(x)$) are two functions of spatial parameter $x$, following identities hold
\begin{equation}\label{Bneg}
(Bf \frac{\Lambda^{-1}}{1-\Lambda^{-1}} g)_-=B(f) \frac{\Lambda^{-1}}{1-\Lambda^{-1}} g,\ \ \ (f \frac{\Lambda^{-1}}{1-\Lambda^{-1}} gB)_-=f \frac{\Lambda^{-1}}{1-\Lambda^{-1}}B^*(g),
\end{equation}
\begin{equation}\label{Cpos}
(Cf \frac{1}{1-\Lambda} g)_+=C(f)\frac{1}{1-\Lambda} g,\ \ \ (f \frac{1}{1-\Lambda} gC)_+=f \frac{1}{1-\Lambda}C^*(g).
\end{equation}
\end{lemma}
\begin{proof}
Here we only give the proof of the eq.\eqref{Bneg} by direct calculation
\[\notag
(Bf \frac{\Lambda^{-1}}{1-\Lambda^{-1}} g)_-&=&\sum_{m=0}^{\infty}b_m(f(x+m\epsilon)\La^m \frac{\Lambda^{-1}}{1-\Lambda^{-1}} g)_-\\ \notag
&=&\sum_{m=0}^{\infty}b_mf(x+m\epsilon)(\frac{\Lambda^{m-1}}{1-\Lambda^{-1}})_- g\\ \notag
&=&\sum_{m=0}^{\infty}b_mf(x+m\epsilon)\frac{\Lambda^{-1}}{1-\Lambda^{-1}} g\\
&=&B(f) \frac{\Lambda^{-1}}{1-\Lambda^{-1}} g,\]

\[\notag
(f \frac{\Lambda^{-1}}{1-\Lambda^{-1}} gB)_-&=&\sum_{m=0}^{\infty}(f \frac{\Lambda^{-1}}{1-\Lambda^{-1}} gb_m\La^m)_-\\ \notag
&=&\sum_{m=0}^{\infty}(f \frac{\Lambda^{-1}}{1-\Lambda^{-1}}\La^m g(x-m\epsilon)b_m(x-m\epsilon))_-\\ \notag
&=&\sum_{m=0}^{\infty}f( \frac{\Lambda^{m-1}}{1-\Lambda^{-1}})_- g(x-m\epsilon)b_m(x-m\epsilon)\\ \notag
&=&f \frac{\Lambda^{-1}}{1-\Lambda^{-1}}B^*(g).\]
The similar proof for the eq.\eqref{Cpos} can be got easily.

\end{proof}

Now, we will give the following important theorem which will be used to generate new solutions from seed solutions.

\begin{theorem}
If $\phi$ is the first wave function of the EMTH,
the Darboux transformation operator of the EMTH  \[W(\lambda)=(1-\phi(\phi(x-\epsilon))^{-1}\La^{-1})=\phi\circ(1-\La^{-1})\circ\phi^{-1},\]

will generater new solutions $u^{[1]},v^{[1]}$ from seed solutions
$u,v$

\[\label{1uN-11}u^{[1]}&=&u+(\La-1)\phi\left(\phi(x-\epsilon)\right)^{-1},\\
\label{1u-M-11} v^{[1]}&=&\phi\left(\phi(x-\epsilon)\right)^{-1}(\La^{-1}v)(\La^{-2}\phi)(\La^{-1}\phi)^{-1}.\]

\end{theorem}
\begin{proof}

In the following proof, using eq.\eqref{Bneg} in the Lemma \ref{lema}, a direct computation will lead to the following
 \begin{eqnarray*}\epsilon W_{t_{\gamma,n}}W^{-1}&=&\left(\phi\circ(1-\La^{-1})\circ\phi^{-1}\right)_{t_{\gamma,n}}\phi\circ(1-\La^{-1})^{-1}\circ\phi^{-1}\\
 &=&\left(((B_{\gamma,n})_+\phi)\circ(1-\La^{-1})\circ\phi^{-1}\right)\phi\circ(1-\La^{-1})^{-1}\circ\phi^{-1}\\
 &&-\phi\circ(1-\La^{-1})\circ\left((B_{\gamma,n})_+\phi\right)\phi^{-1}\circ(1-\La^{-1})^{-1}\circ\phi^{-1}\\
 &=&\left((B_{\gamma,n})_+\phi\right)\phi^{-1}-\phi\circ(1-\La^{-1})\circ\left((B_{\gamma,n})_+\phi\right)\phi^{-1}\circ(1-\La^{-1})^{-1}\circ\phi^{-1}\\
   &=&-\left(\phi\circ[(1-\La^{-1})\cdot\phi^{-1}(x) ((B_{\gamma,n}(x))_+\cdot\phi(x))]\circ(1-\La^{-1})^{-1}\circ\phi^{-1}\right)_-\\
    &=&-\left(\phi\circ(1-\La^{-1})\circ\phi^{-1}(x)\circ (B_{\gamma,n}(x))_+\circ\phi(x)\circ(1-\La^{-1})^{-1}\circ\phi^{-1}\right)_-\\
  &=&-\phi\circ(1-\La^{-1})\circ\phi^{-1}(x)\circ (B_{\gamma,n})_+(x)\circ\phi(x)\circ(1-\La^{-1})^{-1}\circ\phi^{-1}\\
  &&+\left(\phi\circ(1-\La^{-1})\circ\phi^{-1}(x)\circ B_{\gamma,n}(x)\circ\phi(x)\circ(1-\La^{-1})^{-1}\circ\phi^{-1}\right)_+\\
  &=&-W(B_{\gamma,n})_+W^{-1}+(WB_{\gamma,n}W^{-1})_+.
   \end{eqnarray*}
 Therefore  \[W=\phi\circ(1-\La^{-1})\circ\phi^{-1},\]
 can be as a Darboux transformation of the EMTH.
 Eqs.\eqref{1uN-11}-\eqref{1u-M-11} can be directly got from the eq.\eqref{1darbouxL}.

\end{proof}

Here we define $ \phi_i=\phi_i^{[0]}:=\phi|_{\lambda=\lambda_i}$, then one can choose the specific one-fold  Darboux transformation of the EMTH as following
\[W_1(\lambda_1)=\I_N-\phi_1(\phi_1(x-\epsilon))^{-1}\La^{-1}.\]

Meanwhile, we can also get the Darboux transformation on wave function $\phi$ as following

 \[\phi^{[1]}=(\I_N-\phi_1(\phi_1(x-\epsilon))^{-1}\La^{-1})\phi.\]
Then using iteration on the Darboux transformation, the $j$-th Darboux transformation from the $(j-1)$-th solution is as

\[\phi^{[j]}&=&(\I_N-\frac{\phi_j^{[j-1]}}{\La^{-1}\phi_j^{[j-1]}}\La^{-1})\phi^{[j-1]},\\
u^{[j]}&=&u^{[j-1]}+(\La-1)\phi_j^{[j-1]}(\La^{-1}\phi_j^{[j-1]})^{-1},\\
v^{[j]}&=&\phi_j^{[j-1]}(\La^{-1}\phi_j^{[j-1]})^{-1}(\La^{-1} v^{[j-1]})\La^{-2}\phi_j^{[j-1]}(\La^{-1}\phi_j^{[j-1]})^{-1},\]
where $ \phi_i^{[j-1]}:=\phi^{[j-1]}|_{\lambda=\lambda_i},$ are wave functions corresponding to different spectrals with the $(j-1)$-th solutions $u^{[j-1]},v^{[j-1]}.$ It can be checked that $ \phi_i^{[j-1]}=0,\ \ i=1,2,\dots, j-1.$

After iteration on the Darboux transformations, the following theorem about the two-fold  Darboux transformation of the EMTH can be derived by a direct calculation.

\begin{theorem}
The two-fold  Darboux transformation of the EMTH is as following
\[W_2=\I_N+t_1^{[2]}\Lambda^{-1}+t_2^{[2]}\Lambda^{-2},\]
where

 \[ t_1^{[2]}&=&(\phi_{1}\phi_{2}(x-2\epsilon)-\phi_{2}\phi_{1}(x-2\epsilon))
(\phi_{1}(x-\epsilon)\phi_{2}(x-2\epsilon)-\phi_{2}(x-\epsilon)\phi_{1}(x-2\epsilon))^{-1},\notag\\ \\
t_2^{[2]}&=&(\phi_{1}\phi_{2}(x-\epsilon)-\phi_{2}\phi_{1}(x-\epsilon))
(\phi_{2}(x-2\epsilon)\phi_{1}(x-\epsilon)-\phi_{1}(x-2\epsilon)\phi_{2}(x-\epsilon))^{-1}.\notag\\\]

The Darboux transformation leads to new solutions from seed solutions
\[u^{[2]}&=&u+(\La-1)t_1^{[2]},\\
v^{[2]}&=&t_2^{[2]}(x)(\La^{-2}v)t_2^{[2]-1}(x-\epsilon).\]
In fact the Darboux operator $W_2$ can be further written as
\[\notag(W_2)_{ij}&=&\frac{1}{\Delta_2}\left|\begin{matrix}\begin{smallmatrix}
\delta_{ij}&0&\dots & \La^{-1}&\dots & 0&0&\dots & \La^{-2}&\dots & 0\\
\phi_{1,i1}(x)&\phi_{1,11}(x-\epsilon)&\dots & \phi_{1,j1}(x-\epsilon)&\dots &\phi_{1,N1}(x-\epsilon)&\phi_{1,11}(x-2\epsilon)&\dots & \phi_{1,j1}(x-2\epsilon)&\dots & \phi_{1,N1}(x-2\epsilon)\\
\phi_{1,i2}(x)&\phi_{1,12}(x-\epsilon)&\dots & \phi_{1,j2}(x-\epsilon)&\dots & \phi_{1,N2}(x-\epsilon)&\phi_{1,12}(x-2\epsilon)&\dots & \phi_{1,j2}(x-2\epsilon)&\dots & \phi_{1,N2}(x-2\epsilon)\\
\phi_{1,ii}(x)&\phi_{1,1i}(x-\epsilon)&\dots & \phi_{1,ji}(x-\epsilon)&\dots & \phi_{1,Ni}(x-\epsilon)&\phi_{1,1i}(x-2\epsilon)&\dots & \phi_{1,ji}(x-2\epsilon)&\dots & \phi_{1,Ni}(x-2\epsilon)\\
\dots&\dots&\dots&\dots & \dots&\dots &\dots&\dots & \dots&\dots &\dots\\
\phi_{1,iN}(x)&\phi_{1,1N}(x-\epsilon)&\dots & \phi_{1,jN}(x-\epsilon)&\dots & \phi_{1,NN}(x-\epsilon)&\phi_{1,1N}(x-2\epsilon)&\dots & \phi_{1,jN}(x-2\epsilon)&\dots & \phi_{1,NN}(x-2\epsilon)\\
\phi_{2,i1}(x)&\phi_{2,11}(x-\epsilon)&\dots & \phi_{2,j1}(x-\epsilon)&\dots &
\phi_{2,N1}(x-\epsilon)&\phi_{2,21}(x-2\epsilon)&\dots & \phi_{2,j1}(x-2\epsilon)&\dots & \phi_{2,N1}(x-2\epsilon)\\
\phi_{2,i2}(x)&\phi_{2,12}(x-\epsilon)&\dots & \phi_{2,j2}(x-\epsilon)&\dots & \phi_{2,N2}(x-\epsilon)&\phi_{2,12}(x-2\epsilon)&\dots & \phi_{2,j2}(x-2\epsilon)&\dots & \phi_{2,N2}(x-2\epsilon)\\
\phi_{2,ii}(x)&\phi_{2,1i}(x-\epsilon)&\dots & \phi_{2,ji}(x-\epsilon)&\dots & \phi_{2,Ni}(x-\epsilon)&
\phi_{2,1i}(x-2\epsilon)&\dots & \phi_{2,ji}(x-2\epsilon)&\dots & \phi_{2,Ni}(x-2\epsilon)\\
\dots&\dots&\dots&\dots & \dots&\dots &\dots&\dots & \dots&\dots &\dots\\
\phi_{2,iN}&\phi_{2,1N}(x-\epsilon)&\dots & \phi_{2,jN}(x-\epsilon)&\dots & \phi_{2,NN}(x-\epsilon)&\phi_{2,1N}(x-2\epsilon)&\dots & \phi_{2,jN}(x-2\epsilon)&\dots & \phi_{2,NN}(x-2\epsilon)\end{smallmatrix}\end{matrix}
\right|,\]
\[\Delta_2=\left|\begin{matrix}\begin{smallmatrix}
\phi_{1,11}(x-\epsilon)&\dots & \phi_{1,j1}(x-\epsilon)&\dots &\phi_{1,N1}(x-\epsilon)&\phi_{1,11}(x-2\epsilon)&\dots & \phi_{1,j1}(x-2\epsilon)&\dots & \phi_{1,N1}(x-2\epsilon)\\
\phi_{1,12}(x-\epsilon)&\dots & \phi_{1,j2}(x-\epsilon)&\dots & \phi_{1,N2}(x-\epsilon)&\phi_{1,12}(x-2\epsilon)&\dots & \phi_{1,j2}(x-2\epsilon)&\dots & \phi_{1,N2}(x-2\epsilon)\\
\phi_{1,13}(x-\epsilon)&\dots & \phi_{1,j3}(x-\epsilon)&\dots & \phi_{1,N3}(x-\epsilon)&\phi_{1,13}(x-2\epsilon)&\dots & \phi_{1,j3}(x-2\epsilon)&\dots & \phi_{1,N3}(x-2\epsilon)\\
\dots&\dots&\dots & \dots&\dots &\dots&\dots & \dots&\dots &\dots\\
\phi_{1,1N}(x-\epsilon)&\dots & \phi_{1,jN}(x-\epsilon)&\dots & \phi_{1,NN}(x-\epsilon)&\phi_{1,1N}(x-2\epsilon)&\dots & \phi_{1,jN}(x-2\epsilon)&\dots & \phi_{1,NN}(x-2\epsilon)\\
\phi_{2,11}(x-\epsilon)&\dots & \phi_{2,j1}(x-\epsilon)&\dots &
\phi_{2,N1}(x-\epsilon)&\phi_{2,21}(x-2\epsilon)&\dots & \phi_{2,j1}(x-2\epsilon)&\dots & \phi_{2,N1}(x-2\epsilon)\\
\phi_{2,12}(x-\epsilon)&\dots & \phi_{2,j2}(x-\epsilon)&\dots & \phi_{2,N2}(x-\epsilon)&\phi_{2,12}(x-2\epsilon)&\dots & \phi_{2,j2}(x-2\epsilon)&\dots & \phi_{2,N2}(x-2\epsilon)\\
\phi_{2,13}(x-\epsilon)&\dots & \phi_{2,j3}(x-\epsilon)&\dots & \phi_{2,N3}(x-\epsilon)&
\phi_{2,13}(x-2\epsilon)&\dots & \phi_{2,j3}(x-2\epsilon)&\dots & \phi_{2,N3}(x-2\epsilon)\\
\dots&\dots&\dots & \dots&\dots &\dots&\dots & \dots&\dots &\dots\\
\phi_{2,1N}(x-\epsilon)&\dots & \phi_{2,jN}(x-\epsilon)&\dots & \phi_{2,NN}(x-\epsilon)&\phi_{2,1N}(x-2\epsilon)&\dots & \phi_{2,jN}(x-2\epsilon)&\dots & \phi_{2,NN}(x-2\epsilon)\end{smallmatrix}\end{matrix}
\right|.\]
\end{theorem}

Similarly, we can generalize the Darboux transformation to the $n$-fold case which is contained in the following theorem.

\begin{theorem}\label{ndarboux}
The $n$-fold  Darboux transformation of EMTH equation is as follows
\[W_n=\I_N+t_1^{[n]}\Lambda^{-1}+t_2^{[n]}\Lambda^{-2}+\dots+t_{n}^{[n]}\Lambda^{-n}\]
where

\[ W_n\cdot\phi_{i}|_{i\leq n}=0.\]

The Darboux transformation leads to new solutions form seed solutions
\[u^{[n]}&=&u+(\La-1)t_1^{[n]},\\
v^{[n]}&=&t_n^{[n]}(x)(\La^{-n}v)t_n^{[n]-1}(x-\epsilon),\]
where
\[\notag &&(W_n)_{ij}=\frac{1}{\Delta_n}\\ \notag
&&\left|\begin{matrix}\begin{smallmatrix}
\delta_{ij}&0&\dots & \La^{-1}&\dots & 0&0&\dots & 0&\dots & \La^{-n}&\dots & 0\\
\phi_{1,i1}(x)&\phi_{1,11}(x-\epsilon)&\dots & \phi_{1,j1}(x-\epsilon)&\dots &\phi_{1,N1}(x-\epsilon)&\phi_{1,11}(x-2\epsilon)&\dots & \phi_{1,N1}(x-2\epsilon)&\dots & \phi_{1,j1}(x-n\epsilon)&\dots & \phi_{1,N1}(x-n\epsilon)\\
\phi_{1,i2}(x)&\phi_{1,12}(x-\epsilon)&\dots & \phi_{1,j2}(x-\epsilon)&\dots & \phi_{1,N2}(x-\epsilon)&\phi_{1,12}(x-2\epsilon)&\dots & \phi_{1,N2}(x-2\epsilon)&\dots & \phi_{1,j2}(x-n\epsilon)&\dots & \phi_{1,N2}(x-n\epsilon)\\
\phi_{1,ii}(x)&\phi_{1,1i}(x-\epsilon)&\dots & \phi_{1,ji}(x-\epsilon)&\dots & \phi_{1,Ni}(x-\epsilon)&\phi_{1,1i}(x-2\epsilon)&\dots & \phi_{1,Ni}(x-2\epsilon)&\dots & \phi_{1,ji}(x-n\epsilon)&\dots & \phi_{1,Ni}(x-n\epsilon)\\
\dots&\dots&\dots&\dots & \dots&\dots &\dots&\dots &\dots&\dots & \dots&\dots &\dots\\
\phi_{1,iN}(x)&\phi_{1,1N}(x-\epsilon)&\dots & \phi_{1,jN}(x-\epsilon)&\dots & \phi_{1,NN}(x-\epsilon)&\phi_{1,1N}(x-2\epsilon)&\dots& \phi_{1,NN}(x-2\epsilon)&\dots & \phi_{1,jN}(x-n\epsilon)&\dots & \phi_{1,NN}(x-n\epsilon)\\
\phi_{2,i1}(x)&\phi_{2,11}(x-\epsilon)&\dots & \phi_{2,j1}(x-\epsilon)&\dots &
\phi_{2,N1}(x-\epsilon)&\phi_{2,21}(x-2\epsilon)&\dots & \phi_{2,N1}(x-2\epsilon)&\dots & \phi_{2,j1}(x-n\epsilon)&\dots & \phi_{2,N1}(x-n\epsilon)\\
\phi_{2,i2}(x)&\phi_{2,12}(x-\epsilon)&\dots & \phi_{2,j2}(x-\epsilon)&\dots & \phi_{2,N2}(x-\epsilon)&\phi_{2,12}(x-2\epsilon)&\dots & \phi_{2,j2}(x-2\epsilon)&\dots & \phi_{2,N2}(x-n\epsilon)&\dots  & \phi_{2,N2}(x-n\epsilon)\\
\phi_{2,ii}(x)&\phi_{2,1i}(x-\epsilon)&\dots & \phi_{2,ji}(x-\epsilon)&\dots & \phi_{2,Ni}(x-\epsilon)&
\phi_{2,1i}(x-2\epsilon)&\dots  & \phi_{2,Ni}(x-2\epsilon)&\dots & \phi_{2,ji}(x-n\epsilon)&\dots & \phi_{2,Ni}(x-n\epsilon)\\
\dots&\dots&\dots&\dots & \dots&\dots &\dots&\dots & \dots&\dots  & \dots&\dots &\dots\\
\phi_{2,iN}(x)&\phi_{2,1N}(x-\epsilon)&\dots & \phi_{2,jN}(x-\epsilon)&\dots & \phi_{2,NN}(x-\epsilon)&\phi_{2,1N}(x-2\epsilon)&\dots & \phi_{2,NN}(x-2\epsilon)&\dots & \phi_{2,jN}(x-n\epsilon)&\dots & \phi_{2,NN}(x-n\epsilon)\\
\dots&\dots&\dots&\dots & \dots&\dots &\dots&\dots & \dots&\dots  & \dots&\dots &\dots\\
\phi_{n,i1}(x)&\phi_{n,11}(x-\epsilon)&\dots & \phi_{n,j1}(x-\epsilon)&\dots &
\phi_{n,N1}(x-\epsilon)&\phi_{n,21}(x-2\epsilon)&\dots & \phi_{n,N1}(x-2\epsilon)&\dots & \phi_{n,j1}(x-n\epsilon)&\dots & \phi_{n,N1}(x-n\epsilon)\\
\phi_{n,i2}(x)&\phi_{n,12}(x-\epsilon)&\dots & \phi_{n,j2}(x-\epsilon)&\dots & \phi_{n,N2}(x-\epsilon)&\phi_{n,12}(x-2\epsilon)&\dots & \phi_{n,j2}(x-2\epsilon)&\dots & \phi_{n,N2}(x-n\epsilon)&\dots  & \phi_{n,N2}(x-n\epsilon)\\
\phi_{n,ii}(x)&\phi_{n,1i}(x-\epsilon)&\dots & \phi_{n,ji}(x-\epsilon)&\dots & \phi_{n,Ni}(x-\epsilon)&
\phi_{n,1i}(x-2\epsilon)&\dots  & \phi_{n,Ni}(x-2\epsilon)&\dots & \phi_{n,ji}(x-n\epsilon)&\dots & \phi_{n,Ni}(x-n\epsilon)\\
\dots&\dots&\dots&\dots & \dots&\dots &\dots&\dots & \dots&\dots  & \dots&\dots &\dots\\
\phi_{n,iN}(x)&\phi_{n,1N}(x-\epsilon)&\dots & \phi_{n,jN}(x-\epsilon)&\dots & \phi_{n,NN}(x-\epsilon)&\phi_{n,1N}(x-2\epsilon)&\dots & \phi_{n,NN}(x-2\epsilon)&\dots & \phi_{n,jN}(x-n\epsilon)&\dots & \phi_{n,NN}(x-n\epsilon)\end{smallmatrix}\end{matrix}
\right|,\]
\[&&\notag \Delta_n=\\ \notag
&&\left|\begin{matrix}\begin{smallmatrix}
\phi_{1,11}(x-\epsilon)&\phi_{1,21}(x-\epsilon)&\dots &\phi_{1,N1}(x-\epsilon)&\phi_{1,11}(x-2\epsilon)&\phi_{1,21}(x-2\epsilon)&\dots & \phi_{1,N1}(x-2\epsilon)&\phi_{1,21}(x-n\epsilon)&\dots & \phi_{1,N1}(x-n\epsilon)\\
\phi_{1,12}(x-\epsilon)&\phi_{1,22}(x-\epsilon)&\dots & \phi_{1,N2}(x-\epsilon)&\phi_{1,12}(x-2\epsilon)&\phi_{1,22}(x-2\epsilon)&\dots & \phi_{1,N2}(x-2\epsilon)&\phi_{1,22}(x-n\epsilon)&\dots & \phi_{1,N2}(x-n\epsilon)\\
\phi_{1,1i}(x-\epsilon)&\phi_{1,2i}(x-\epsilon)&\dots & \phi_{1,Ni}(x-\epsilon)&\phi_{1,1i}(x-2\epsilon)&\phi_{1,2i}(x-2\epsilon)&\dots & \phi_{1,Ni}(x-2\epsilon)&\phi_{1,2i}(x-n\epsilon)&\dots & \phi_{1,Ni}(x-n\epsilon)\\
\dots&\dots&\dots & \dots&\dots &\dots&\dots & \dots&\dots &\dots&\dots \\
\phi_{1,1N}(x-\epsilon)&\phi_{1,2N}(x-\epsilon)&\dots & \phi_{1,NN}(x-\epsilon)&\phi_{1,1N}(x-2\epsilon)&\phi_{1,2N}(x-2\epsilon)&\dots & \phi_{1,NN}(x-2\epsilon)&\phi_{1,2N}(x-n\epsilon)&\dots & \phi_{1,NN}(x-n\epsilon)\\
\phi_{2,11}(x-\epsilon)&\phi_{2,21}(x-\epsilon)&\dots &
\phi_{2,N1}(x-\epsilon)&\phi_{2,21}(x-2\epsilon)&\phi_{2,21}(x-2\epsilon)&\dots & \phi_{2,N1}(x-2\epsilon)&\phi_{2,21}(x-n\epsilon)&\dots & \phi_{2,N1}(x-n\epsilon)\\
\phi_{2,12}(x-\epsilon)&\phi_{2,22}(x-\epsilon)&\dots & \phi_{2,N2}(x-\epsilon)&\phi_{2,12}(x-2\epsilon)&\phi_{2,22}(x-2\epsilon)&\dots & \phi_{2,N2}(x-2\epsilon)&\phi_{2,22}(x-n\epsilon)&\dots & \phi_{2,N2}(x-n\epsilon)\\
\phi_{2,1i}(x-\epsilon)&\phi_{2,2i}(x-\epsilon)&\dots & \phi_{2,Ni}(x-\epsilon)&
\phi_{2,1i}(x-2\epsilon)&\phi_{2,2i}(x-2\epsilon)&\dots & \phi_{2,Ni}(x-2\epsilon)&\phi_{2,2i}(x-n\epsilon)&\dots & \phi_{2,Ni}(x-n\epsilon)\\
\dots&\dots&\dots &\dots & \dots&\dots &\dots&\dots & \dots&\dots &\dots\\
\phi_{2,1N}(x-\epsilon)&\phi_{2,2N}(x-\epsilon)&\dots & \phi_{2,NN}(x-\epsilon)&\phi_{2,1N}(x-2\epsilon)&\phi_{2,2N}(x-2\epsilon)&\dots & \phi_{2,NN}(x-2\epsilon)&\phi_{2,2N}(x-n\epsilon)&\dots & \phi_{2,NN}(x-n\epsilon)\\
\dots&\dots&\dots &\dots & \dots&\dots &\dots&\dots & \dots&\dots &\dots\\
\phi_{n,11}(x-\epsilon)&\phi_{n,21}(x-\epsilon)&\dots &
\phi_{n,N1}(x-\epsilon)&\phi_{n,21}(x-2\epsilon)&\phi_{n,21}(x-2\epsilon)&\dots & \phi_{n,N1}(x-2\epsilon)&\phi_{n,21}(x-n\epsilon)&\dots & \phi_{n,N1}(x-n\epsilon)\\
\phi_{n,12}(x-\epsilon)&\phi_{n,22}(x-\epsilon)&\dots & \phi_{n,N2}(x-\epsilon)&\phi_{n,12}(x-2\epsilon)&\phi_{n,22}(x-2\epsilon)&\dots & \phi_{n,N2}(x-2\epsilon)&\phi_{n,22}(x-n\epsilon)&\dots & \phi_{n,N2}(x-n\epsilon)\\
\phi_{n,1i}(x-\epsilon)&\phi_{n,2i}(x-\epsilon)&\dots & \phi_{n,Ni}(x-\epsilon)&
\phi_{n,1i}(x-2\epsilon)&\phi_{n,2i}(x-2\epsilon)&\dots & \phi_{n,Ni}(x-2\epsilon)&\phi_{n,2i}(x-n\epsilon)&\dots & \phi_{n,Ni}(x-n\epsilon)\\
\dots&\dots&\dots &\dots & \dots&\dots &\dots&\dots & \dots&\dots &\dots\\
\phi_{n,1N}(x-\epsilon)&\phi_{n,2N}(x-\epsilon)&\dots & \phi_{n,NN}(x-\epsilon)&\phi_{n,1N}(x-2\epsilon)&\phi_{n,2N}(x-2\epsilon)&\dots & \phi_{n,NN}(x-2\epsilon)&\phi_{n,2N}(x-n\epsilon)&\dots & \phi_{n,NN}(x-n\epsilon)\end{smallmatrix}\end{matrix}
\right|.\]
\end{theorem}
It can be easily checked that $W_n\phi_i=0,\ i=1,2,\dots,n.$

\section{bi-Hamiltonian structure and tau symmetry}
To describe the integrability of the EMTH with the Lax operator
 \[\L=\Lambda+u+v\Lambda^{-1},\] we will construct the bi-Hamiltonian structure and tau symmetry of the EMTH in this section.
For a matrix $A=(a_{ij})$, the vector field $\d_A$ over EMTH is defined by
\[\d_A=\sum_{i,j=1}^N\sum_{k\geq 0}a_{ij}^{(k)}\left(\frac{\d}{\d u_{ij}^{(k)}}+\frac{\d}{\d v_{ij}^{(k)}}\right)=Tr \sum_{k\geq 0}A^{(k)}\left(\frac{\d}{\d u^{(k)}}+\frac{\d}{\d v^{(k)}}\right),\]
where
\[\left(\frac{\d}{\d u^{(k)}}\right)_{ji}=\frac{\d}{\d u_{ij}^{(k)}},\ \left(\frac{\d}{\d v^{(k)}}\right)_{ji}=\frac{\d}{\d v_{ij}^{(k)}}.\]

For two functionals $\bar f=\int f dx,\bar g=\int g dx$, we have
\[\d_A \bar f=\int\sum_{i,j=1}^N\sum_{k\geq 0}a_{ij}^{(k)}\left(\frac{\d f}{\d u_{ij}^{(k)}}+\frac{\d f}{\d v_{ij}^{(k)}}\right) dx=\int Tr \sum_{k\geq 0}A^{(k)}\left(\frac{\delta f}{\delta u^{(k)}}+\frac{\delta f}{\delta v^{(k)}}\right) dx.\]
Then we can define the Hamiltonian bracket as
\[\{\bar f,\bar g\}=\int \sum_{w,w'}\frac{\delta f}{\delta w}\{w,w'\}\frac{\delta g}{\delta w'} dx,\ \ w, w'=u_{ij}\ or\ v_{ij}, \ 1\leq i,j \leq N.\]
The bi-Hamiltonian structure for the
EMTH can be given by the following two compatible Poisson brackets which is a generalization in matrix forms of the extended Toda hierarchy in \cite{CDZ}

\begin{eqnarray}
\{u(x)_{ij},u(y)_{pq}\}_1&=&\frac{1}{\epsilon} [\delta_{iq}u_{pj}(x)-\delta_{jp}u_{iq}(x)]\delta(x-y),\label{toda-pb1uu}\\
\{u(x)_{ij},v(y)_{pq}\}_1&=&\frac{1}{\epsilon} \left[\delta_{iq}\Lambda v_{pj}(x)-\delta_{jp}v_{iq}(x)\right]\delta(x-y),\label{toda-pb1uv}\\
\{v(x)_{ij},v(y)_{pq}\}_1&=&0,\label{toda-pb1vv}\\ \notag
\{u(x)_{ij},u(y)_{pq}\}_2&=&\frac{1}{\epsilon} \left[\delta_{iq}\Lambda v_{pj}(x)-\delta_{jp}v_{iq}(x)\Lambda^{-1}+\delta_{iq}\sum_{s=1}^Nu_{sj}\frac{\La}{\La-1}u_{ps}-u_{pj}\frac{\La}{\La-1}u_{iq}\right.\\
&&\left.-u_{iq}(\La-1)^{-1}u_{pj}+\delta_{jp}\sum_{s=1}^Nu_{is}(\La-1)^{-1}u_{sq}\right]\delta(x-y), \label{toda-pb2uu}\\ \notag
 \{ u(x)_{ij}, v(y)_{pq}\}_2 &=& \frac{1}{\epsilon} \left[\delta_{iq}\sum_{s=1}^Nu_{sj}\La^2 (\La-1)^{-1} v_{ps}(x)-u_{pj}\La (\La-1)^{-1}v_{iq}\right.\\
&&\left.-u_{iq}(\La-1)^{-1}\La v_{pj}(x)+\delta_{jp}\sum_{s=1}^Nu_{is}(\La-1)^{-1}v_{sq}\right]
\delta(x-y),\label{toda-pb2uv}\\ \notag
 \{ v(x)_{ij}, v(y)_{pq}\}_2 &=& {1\over \epsilon} \left[\delta_{iq}\sum_{s=1}^Nv_{sj}\La^2 (\La-1)^{-1}v_{ps}(x)-v_{pj}\La (\La-1)^{-1}v_{iq}\right.\\
&&\left.-v_{iq}(\La-1)^{-1}v_{pj}(x)+\delta_{jp}\sum_{s=1}^Nv_{is}\La^{-1}(\La-1)^{-1}v_{sq}\right]\delta(x-y). \label{toda-pb2vv}
\end{eqnarray}

One can check the anti-symmetric property of the above Poisson bracket, i.e.
\[\{u_{ij},u_{pq}\}_i=-\{u_{pq},u_{ij}\}_i,\ \ \{v_{ij},v_{pq}\}=-\{v_{pq},v_{ij}\}_i,\ \ i=1,2,\]
and so on.
Further one can prove the two Poisson  brackets all satisfying the anti-symmetric property.
Also one can prove the complicated Jacobi identities over the two Hamiltonian structures about which the detail will not be given here explicitly.
Also the compatibility of the two Hamiltonian structures can be got after the following theorem which shows the bi-Hamiltonian recursion relation.

When $N=1$, the bi-Hamiltonian structure will be reduced to the following bracket
\begin{eqnarray}
\{u(x),u(y)\}_1&=&\{v(x),v(y)\}_1=0,\label{1toda-pb1uu}\\
\{u(x),v(y)\}_1&=&\frac{1}{\epsilon} \left[\Lambda v(x)-1\right]\delta(x-y),\label{1toda-pb1uv}\\ \notag
\{u(x),u(y)\}_2&=&\frac{1}{\epsilon} \left[\Lambda v(x)-v(x)\Lambda^{-1}\right]\delta(x-y), \label{1toda-pb2uu}\\ \notag
 \{ u(x), v(y)\}_2 &=& \frac{1}{\epsilon} \left[u(x)(\La-1) v(x)\right]
\delta(x-y),\label{1toda-pb2uv}\\ \notag
 \{ v(x), v(y)\}_2 &=& {1\over \epsilon} \left[v(x)\La v(x)-v(x)\La^{-1}v(x)\right]\delta(x-y), \label{1toda-pb2vv}
\end{eqnarray}
which can be rewritten as

\begin{eqnarray}
\{u(x),u(y)\}_1&=&\{\log v(x),\log v(y)\}_1=0,\label{1toda-pb1uu}\\
\{u(x),\log v(y)\}_1&=&\frac{1}{\epsilon} \left[\Lambda-1\right]\delta(x-y),\label{1toda-pb1uv}\\ \notag
\{u(x),u(y)\}_2&=&\frac{1}{\epsilon} \left[\Lambda v(x)-v(x)\Lambda^{-1}\right]\delta(x-y), \label{1toda-pb2uu}\\ \notag
 \{ u(x), \log  v(y)\}_2 &=& \frac{1}{\epsilon} \left[u(x)(\La-1) \right]
\delta(x-y),\label{1toda-pb2uv}\\ \notag
 \{ \log v(x), \log v(y)\}_2 &=& {1\over \epsilon} \left[\La -\La^{-1}\right]\delta(x-y). \label{1toda-pb2vv}
\end{eqnarray}
In the above computation, the following identity is used
\[\La [v(x)\delta(x-y)]=v(x+\epsilon)\delta(x+\epsilon-y)=v(y)\delta(x+\epsilon-y)=v(y)\La\delta(x-y).\]
This is exactly the bi-Hamiltonian structure of the extended Toda hierarchy in \cite{CDZ} if we rewrite the $\log v$ equivalently to a new function $u$ in \cite{CDZ}.

For any difference operator $A=
\sum_k A_k \Lambda^k$, we define its residue as $Res\  A=A_0$.
In the following theorem, we will prove the above Poisson structures can be considered as the bi-Hamiltonian structure of the METH.
\begin{theorem}
The flows of the EMTH  are Hamiltonian systems
of the form
\[
\frac{\d u_{pq}}{\d t_{j,k}}&=&\{u_{pq},H_{j,k}\}_1, \  \frac{\d v_{pq}}{\d t_{j,k}}=\{v_{pq},H_{j,k}\}_1,\\
\frac{\d u_{pq}}{\d \bar t_{j,k}}&=&\{u_{pq},\bar H_{j,k}\}_1, \  \frac{\d v_{pq}}{\d\bar t_{j,k}}=\{v_{pq},\bar H_{j,k}\}_1,\\
\frac{\d u_{pq}}{\d s_{j}}&=&\{u_{pq},\tilde H_{j}\}_1, \  \frac{\d v_{pq}}{\d s_{j}}=\{v_{pq},\tilde H_{j}\}_1,
\quad k=0,1,\dots N;\ j\ge 0.
\label{td-ham}
\]
They satisfy the following bi-Hamiltonian recursion relation
\[
\{\cdot,H_{n-1,k}\}_2&=&\{\cdot,H_{n,k}\}_1,\ \{\cdot,\bar H_{n-1,k}\}_2=
\{\cdot,\bar H_{n,k}\}_1,\\ \label{recursion}
\{\cdot,\tilde H_{n-1}\}_2&=&n
\{\cdot,\tilde H_{n}\}_1+\frac{2}{n!}\sum_{k=1}^N\{\cdot,H_{n,k}\}_1.
\]
Here the Hamiltonians have the form
\begin{equation}
F_{j,k}=\int f_{j,k}(u,v; u_x,v_x; \dots; \epsilon) dx,
\end{equation}
with the Hamiltonian $F_{j,k}=H_{j,k},\bar H_{j,k},\tilde H_{j}$ and the Hamiltonian densities $f_{j,k}=h_{j,k},\bar h_{j,k},\tilde h_{j}$ given by
\[
 h_{j,k}&=&Tr Res \, C_{kk}\L^{j},\ \bar h_{j,k}=Tr Res \, \bar C_{kk} \L^{j},\\
  \tilde h_{j}&=&\frac2{j!}\,Tr Res\left[ \L^{j}
(\log \L-c_{j})\right].
\]

\end{theorem}

\begin{proof}
Here we will only prove that the flows $\frac{\d}{\d s_{n}}$ are
Hamiltonian systems with respect to the first Poisson bracket and their recursion relation. The other cases can be proved in a more simple way.
In \cite{CDZ}, the following identity has been proved
\begin{equation}\label{dlgl-2} Tr Res\left[\L^n d (S\epsilon \d_x
S^{-1})\right] \sim Tr Res \L^{n-1} d \L,
\end{equation}
which show the validity of the following equivalence relation:
\begin{equation}\label{dlgl}
Tr Res\left(\L^n\, d \log_+ \L\right) \sim Tr Res\left(\L^{n-1} d \L\right).
\end{equation}
Here the equivalent relation $\sim$ is up to a $x$-derivative of
another 1-form.

In a similar way as eq.\eqref{dlgl-2}, we obtain the following equivalence relation
\begin{equation}\label{dlgl-3}
Tr Res\left[\L^n d (\bar S\epsilon\d_x \bar S^{-1})\right]\sim - Tr Res \L^{n-1} d \L,
\end{equation}
i.e.
\begin{equation}\label{dlgl}
Tr Res\left(\L^n\, d \log_- \L\right) \sim  Tr Res\left(\L^{n-1} d \L\right),
\end{equation}
which further leads to
\begin{equation}\label{dlgl}
Tr Res\left(\L^n\, d \log\L\right) \sim   Tr  Res\left(\L^{n-1} d \L\right).
\end{equation}
We suppose
\[
D_{n}=\sum_{k} a_{n;k}\, \Lambda^k,\ \ \L^n=\sum_{j}b_{n;j}\,  \Lambda^j,
\]
then from
\begin{equation}
  \label{edef3}
\epsilon \partial_{s_{n}} \L = [ (D_{n})_+ ,\L ]= [ -(D_{n})_- ,\L ],
\end{equation}
we can derive equations
\[\epsilon\frac{\partial u}{\partial s_{n}}&=&a_{n;1}(x)v(x+\epsilon)-v(x)a_{n;1}(x-\epsilon)+[a_{n,0}(x),u(x)],\\
\epsilon\frac{\partial v}{\partial s_{n}}&=&a_{n;0}(x)v(x)-v(x)a_{n;0}(x-\epsilon) .
\]

The equivalence relation (\ref{dlgl}) now readily follows from the above two equations.
By using (\ref{dlgl}) we obtain
\begin{eqnarray}
&&d \tilde h_{n}=\frac2{n!}\,d\,Tr Res\left[\L^{n}
\left(\log \L-c_{n}\right) \right]
\notag\\
&& \sim \frac2{(n-1)!}\,Tr Res\left[\L^{n-1}
\left(\log \L-c_{n}\right) d \L\right]+ \frac2{n!}\,Tr Res\left[\L^{n-1} d \L\right]\notag\\
&&=\frac2{(n-1)!}\,Tr Res\left[\L^{n-1} \left(\log \L-c_{n-1}\right) d \L\right]\\
&&=Tr \left[a_{n-1;0}(x)du+a_{n-1;1}(x) dv(x+\epsilon)\right].
\end{eqnarray}
Then we have
\begin{eqnarray}
&&d \tilde H_{n}=\frac2{n!}\,d\,\int Tr Res\left[\L^{n}
\left(\log \L-c_{n}\right) \right]
\notag\\
&& \sim \int Tr \left[a_{n-1;0}(x)du+a_{n-1;1}(x) dv(x+\epsilon)\right]\\
&& := \int Tr \left[\frac{\delta \tilde H_{n}}{\delta u}du+\frac{\delta \tilde H_{n}}{\delta v}dv\right].
\end{eqnarray}
It yields the following identities
\begin{equation}\label{dH1-u12}
\frac{\delta \tilde H_{n}}{\delta u_{ij}}=a_{n-1;0}(x)_{ji},\quad \frac{\delta \tilde H_{n}}
{\delta v_{ij}}=a_{n-1;1}(x-\epsilon)_{ji}.
\end{equation}
This agrees with the Lax equation

\[\notag
\frac{\d u_{ij}}{\d s_{n}}&=&\{u_{ij},\tilde H_{n+1}\}_1\\ \notag
&=&\sum_{p,q=1}^N\frac{1}{\epsilon} [\delta_{iq}u_{pj}(x)-\delta_{jp}u_{iq}(x)]\frac{\delta \tilde H_{n+1}}
{\delta u_{pq}}+\sum_{p,q=1}^N\frac{1}{\epsilon} \left[\delta_{iq}\Lambda v_{pj}(x)-\delta_{jp}v_{iq}(x)\right]\frac{\delta \tilde H_{n+1}}
{\delta v_{pq}}\\
&=&{1\over \epsilon}[\sum_{p=1}^Na_{n,0}(x)_{ip}u_{pj}(x)-\sum_{q=1}^Nu_{iq}(x)a_{n,0}(x)_{qj}]\notag\\ \notag
&&+{1\over \epsilon}[\sum_{p=1}^Na_{n;1}(x)_{ip}v(x+\epsilon)_{pj}-\sum_{q=1}^Nv(x)_{iq}a_{n;1}(x-\epsilon)_{qj}]\\
&=&{1\over \epsilon}[a_{n,0}(x),u(x)]_{ij}+{1\over \epsilon}[a_{n;1}(x)v(x+\epsilon)-v(x)a_{n;1}(x-\epsilon)]_{ij},\\
 \notag\  \frac{\d v_{ij}}{\d s_{n}}&=&\{v_{ij},\tilde H_{n+1}\}_1=-\sum_{p,q=1}^N\frac{1}{\epsilon} \left[\delta_{pj} v_{iq}(x)\Lambda^{-1}-\delta_{qi}v_{pj}(x)\right]\frac{\delta \tilde H_{n+1}}
{\delta u_{pq}}\\ \notag
&=&{1\over \epsilon}[\sum_{p=1}^Na_{n;0}(x)_{ip}v(x)_{pj}-\sum_{q=1}^Nv(x)_{iq}a_{n;0}(x-\epsilon)_{qj}]\\
&=&\frac{1}{\epsilon} \left[a_{n;0}(x)v(x)-v(x)a_{n;0}(x-\epsilon)\right]_{ij}.
\]

 From the above identities we see that
the flows $\frac{\d}{\d s_{n}}$ are Hamiltonian systems
of the form (\ref{td-ham}).
The recursion relation (\ref{recursion})
follows from the following trivial identities
\begin{eqnarray}
&&n\, \frac{2}{n!} \L^{n} \left(\log \L-c_{n}\right)=\L\,
\frac{2}{(n-1)!}
\L^{n-1} \left(\log \L-c_{n-1}\right)-2\,\frac1{n!} \L^n\notag\\
&&=\frac{2}{(n-1)!} \L^{n-1} \left(\log \L-c_{n-1}\right)\,
\L-2\,\frac1{n!} \L^n.\notag
\end{eqnarray}
Then we get,
\begin{eqnarray}
&&n a_{n;0}(x)=a_{n-1;-1}(x+\epsilon)+u(x)a_{n-1;0}(x)+v(x)a_{n-1;1}(x-\epsilon)-\frac{2}{n!}b_{n;0}(x)\notag\\
&&=a_{n-1;-1}(x)+a_{n-1;0}(x)u(x)+a_{n-1;1}(x)v(x+\epsilon)-\frac{2}{n!}b_{n;0}(x),\label{an0}
\end{eqnarray}
\begin{eqnarray}
&&n a_{n;1}(x)=a_{n-1;0}(x+\epsilon)+u(x)a_{n-1;1}(x)+v(x)a_{n-1;2}(x-\epsilon)-\frac{2}{n!}b_{n;1}(x)\notag\\
&&=a_{n-1;0}(x)+a_{n-1;1}(x)u(x+\epsilon)+a_{n-1;2}(x)v(x+2\epsilon)-\frac{2}{n!}b_{n;1}(x).
\end{eqnarray}

From eq.\eqref{an0}, we can derive the following equation by which one can represent $a_{n-1;-1}(x)$ in terms of other functions
\[\notag
(\La-1)a_{n-1;-1}(x)&=&a_{n-1;1}(x) v(x+\epsilon)-v(x)  a_{n-1;1}(x-\epsilon)+
[ a_{n-1;0}(x),u(x) ].\]
This further leads to the following recursion relation between two Poisson brackets

\begin{eqnarray*}
&&\{u_{ij},\tilde H_{n}\}_2\\
&=&\sum_{p,q,s=1}^N\frac{1}{\epsilon} \left[\delta_{iq}\Lambda v_{pj}(x)-\delta_{jp}v_{iq}(x)\Lambda^{-1}+\delta_{iq}u_{sj}\La (\La-1)^{-1}u_{ps}-u_{pj}\La (\La-1)^{-1}u_{iq}\right.\\
&&\left.-u_{iq}(\La-1)^{-1}u_{pj}+\delta_{jp}u_{is}(\La-1)^{-1}u_{sq}\right]a_{n-1;0}(x)_{qp}\\
&&+\sum_{p,q,s=1}^N\frac{1}{\epsilon} \left[\delta_{iq}u_{sj}\La (\La-1)^{-1}v_{ps}(x+\epsilon)\La-u_{pj}\La (\La-1)^{-1}v_{iq}\right.\\
&&\left.-u_{iq}(\La-1)^{-1}v_{pj}(x+\epsilon)\La+\delta_{jp}u_{is}(\La-1)^{-1}v_{sq}\right]a_{n-1;1}(x-\epsilon)_{qp}\\
&=&\frac{1}{\epsilon} \{a_{n-1;0}(x+\epsilon) v(x+\epsilon)-v(x)  a_{n-1;0}(x-\epsilon)\notag\\
\notag && +\La (\La-1)^{-1}(a_{n-1;1}(x) v(x+\epsilon)-v(x)  a_{n-1;1}(x-\epsilon)+
[ a_{n-1;0}(x),u(x) ]) u(x)\\
&&-u(x)(\La-1)^{-1}(a_{n-1;1}(x) v(x+\epsilon)-v(x)  a_{n-1;1}(x-\epsilon)+
[ a_{n-1;0}(x),u(x) ])\}_{ij}\\ \notag
&=&\frac{1}{\epsilon} \{a_{n-1;0}(x+\epsilon) v(x+\epsilon)-v(x)  a_{n-1;0}(x-\epsilon)+
u(x) a_{n-1;1}(x) v(x+\epsilon)-v(x) a_{n-1;1}(x-\epsilon)u(x)\notag\\
\notag && + a_{n-1;-1}(x+\epsilon) u(x)-u(x)a_{n-1;-1}(x) +v(x)a_{n-1;1}(x-\epsilon)u(x)-u(x)a_{n-1;1}(x)v(x+\epsilon)\}_{ij}\\ \notag
&=&\frac{n}{\epsilon} \left[a_{n;1}(x) v(x+\epsilon)-v(x)a_{n;1}(x-\epsilon) +[a_{n;0}(x), u(x) ]\right]_{ij}\\
&&+\frac{2}{\epsilon n!}\left[b_{n;1}(x) v(x+\epsilon)- v(x)b_{n;1}(x-\epsilon)+[b_{n;0}(x), u(x) ]\right]_{ij}\\
&=&n
\{u_{ij},\tilde H_{n}\}_1+\frac{2}{n!}\sum_{k=1}^N\{u_{ij},H_{n,k}\}_1.\label{pre-recur}
\end{eqnarray*}
This is exactly the recursion relation eq.\eqref{recursion} for matrix $u$.
 The similar recursion flow on the matrix function $v$ can be similarly derived by the following calculation,
\begin{eqnarray*}
&&\{v_{ij},\tilde H_{n}\}_2\\
&=&\sum_{p,q,s=1}^N\frac{1}{\epsilon} \left[\delta_{iq}v_{sj}\La^2 (\La-1)^{-1}v_{ps}(x)-v_{pj}\La (\La-1)^{-1}v_{iq}\right.\\
&&\left.-v_{iq}(\La-1)^{-1}v_{pj}(x)+\delta_{jp}v_{is}\La^{-1}(\La-1)^{-1}v_{sq}\right]a_{n-1;1}(x-\epsilon)_{qp}\\
&&+\sum_{p,q,s=1}^N\frac{1}{\epsilon} \left[u_{iq}(x)v_{pj}(x) - v_{iq}(x)\La^{-1}u_{pj}(x)+\delta_{iq}v_{sj}(x)\La (\La-1)^{-1}u_{ps}\right.\\
&&\left.-v_{pj}(x)\La (\La-1)^{-1}u_{iq}- v_{iq}\La^{-1} (\La-1)^{-1}u_{pj}+\delta_{pj}v_{is}\La^{-1} (\La-1)^{-1}u_{sq}\right]a_{n-1;0}(x)_{qp}\\
&=&\frac{1}{\epsilon} \{u(x) a_{n-1;0}(x) v(x)-v(x) a_{n-1;0}(x-\epsilon)u(x-\epsilon) \notag\\
\notag && +[\La (\La-1)^{-1}(a_{n-1;1}(x) v(x+\epsilon)-v(x)  a_{n-1;1}(x-\epsilon)+
[ a_{n-1;0}(x),u(x) ])] v(x)\\
&&-v(x)\La^{-1}(\La-1)^{-1}(a_{n-1;1}(x) v(x+\epsilon)-v(x)  a_{n-1;1}(x-\epsilon)+
[ a_{n-1;0}(x),u(x) ])\}_{ij}\\ \notag
&=&\frac{1}{\epsilon} \{u(x) a_{n-1;0}(x) v(x)-v(x) a_{n-1;0}(x-\epsilon)u(x-\epsilon) + a_{n-1;-1}(x+\epsilon) v(x)-v(x)a_{n-1;-1}(x-\epsilon) \}_{ij}\\ \notag
&=&\frac{n}{\epsilon} \left[a_{n;0}(x) v(x)-v(x)a_{n;0}(x-\epsilon)\right]_{ij}+\frac{2}{\epsilon n!}\left[b_{n;0}(x) v(x)- v(x)b_{n;0}(x-\epsilon)\right]_{ij}\\
&=&n
\{v_{ij},\tilde H_{n}\}_1+\frac{2}{n!}\sum_{k=1}^N\{v_{ij},H_{n,k}\}_1.\label{pre-recur}
\end{eqnarray*}
The theorem is proved till now.

\end{proof}

For readers' convenience, now we will write down the first several Hamiltonian densities explicitly as follows
\[ h_{0,k}&=&Tr Res\,  C_{kk}=Tr E_{kk}=1,\\
 h_{1,k}&=&Tr Res \, C_{kk}\L=Tr [(1-\La)^{-1}uE_{kk}-E_{kk}\frac{\La}{1-\La}u]=u_{kk},\\
 \ \bar h_{0,k}&=&Tr Res\,\bar C_{kk} =Tr\, \tilde \omega_0E_{kk} \tilde \omega_0^{-1}  \,  ,\\
 \ \bar h_{1,k}&=&Tr Res \, \bar C_{kk} \L=Tr\, (\tilde \omega_1E_{kk}\tilde \omega_0^{-1}-\tilde \omega_0E_{kk}\tilde \omega_0^{-1}(x-\epsilon)\tilde \omega_1(x-\epsilon)\tilde \omega_0^{-1}),\\
  \tilde h_{0}&=&2\,Tr Res \log \L=-Tr\, \tilde \omega_{0x}\tilde \omega_0^{-1},\\
  \tilde h_{1}&=&\-\,Tr\, Res\left[ \L
(\log \L-1)\right]\\
&=&Tr\,[u-(1-\La)^{-1}u_x-u\tilde \omega_{0x}\tilde \omega_0^{-1}-v\tilde \omega_{0x}(x-\epsilon)\tilde \omega_0^{-1}-\tilde \omega_{0x}(x-\epsilon)\tilde \omega_0^{-1}(x-\epsilon)\tilde \omega_1(x-\epsilon)\tilde \omega_0^{-1}],\notag
\]
with
\[\tilde \omega_0=v\tilde \omega_0(x-\epsilon),\ \tilde \omega_1=u\tilde \omega_0+v\tilde \omega_1(x-\epsilon).\]
When $N=1$, the above conserved densities will be the ones of the extended Toda  hierarchy in \cite{CDZ}.
Similarly as \cite{CDZ}, the tau symmetry of the METH can be proved in the  following theorem.
\begin{theorem}\label{tausymmetry}
The Hamiltonian densities  of the EMTH have the following tau-symmetry property:
\begin{equation}
\frac{\d h_{\alpha,m}}{\d t_{j,k}}=\frac{\d
h_{j,k}}{\d t_{\alpha,m}},\quad \frac{\d \bar h_{\alpha,m}}{\d t_{j,k}}=\frac{\d
 h_{j,k}}{\d\bar t_{\alpha,m}},
\end{equation}
\begin{equation}
\frac{\d h_{\alpha,m}}{\d \bar t_{j,k}}=\frac{\d
\bar h_{j,k}}{\d  t_{\alpha,m}},\quad \frac{\d \bar h_{\alpha,m}}{\d \bar t_{j,k}}=\frac{\d
\bar h_{j,k}}{\d \bar t_{\alpha,m}},
\end{equation}
\begin{equation}\label{tildeht}
\frac{\d \tilde h_{m}}{\d  t_{j,k}}=\frac{\d
h_{j,k}}{\d s_{m}},\quad \frac{\d \tilde h_{m}}{\d \bar t_{j,k}}=\frac{\d
\bar h_{j,k}}{\d s_{m}}.
\end{equation}
\end{theorem}
\begin{proof} Let us prove the theorem for the  first equation in eqs.\eqref{tildeht},
other cases can be proved in a similar way
\[
&&\frac{\d \tilde h_{m}}{\d t_{n,k}} =\frac2{m!\,}\,Tr Res[-(C_{kk}\L^{n})_-, \L^m (\log \L-c_m)]\notag\\
&&=\frac2{m!\,}\,Tr Res[(\L^m (\log \L-c_m))_+,(C_{kk}\L^{n})_-]\notag\\
&& =\frac2{m!\,}\,Tr Res[(\L^m (\log \L
-c_m))_+,C_{kk}\L^{n}]=\frac{\d h_{n,k}}{\d s_{m}}.
\]
The theorem is proved.
\end{proof}

 This property justifies the following definition of the
tau function for the EMTH:

\begin{definition} The  $tau$ function $\tau$ of the EMTH can be defined by
the following expressions in terms of the densities of the Hamiltonians:
\begin{equation}
h_{j,n}=\epsilon (\Lambda-1)\frac{\d\log\tau}{\d t_{j,n}},
\end{equation}
\begin{equation}
\bar h_{j,n}=\epsilon (\Lambda-1)\frac{\d\log\tau}{\d\bar t_{j,n}},
\end{equation}
\begin{equation}
h_{j}=\epsilon (\Lambda-1)\frac{\d\log\tau}{\d s_{j}}.
\end{equation}
\end{definition}

With above two different definitions on tau functions of this hierarchy, some mysterious connections between these two tau functions become an open interesting subject. One comes from the wave function and another comes from Hamiltonians.  This is not easy and will be included in our future work.

\section{Conclusions and Discussions}
In this paper, we constructed a new hierarchy called the EMTH and further extended the Sato theory to
this hierarchy including Sato equations, matrix wave operators, Hirota
quadratic equations, the existence of
the tau function.  Similarly as extended Toda hierarchy and extended bigraded Toda hierarchy in Gromov-Witten theory of $CP^1$,  this
hierarchy deserves further studying and exploring because of its
potential applications in topological quantum fields  and
Gromov-Witten theory. Because the matrix hierarchy is one special important noncommutative integrable system, what is applications of the EMTH in noncommutative geometry  becomes an interesting subject.

{\bf {Acknowledgements:}}
  Chuanzhong Li is supported by the National Natural Science Foundation of China under Grant No. 11201251, Zhejiang Provincial Natural Science Foundation of China under Grant No. LY12A01007, the Natural Science Foundation of Ningbo under Grant No. 2013A610105. Jingsong He is supported by the National Natural Science Foundation of China under Grant No. 11271210, K.C.Wong Magna Fund in
Ningbo University. Chuanzhong Li would like to thank Professor Todor E Milanov for his valuable discussion.

\end{document}